\documentclass[conference]{vldb}
\usepackage{amsfonts}
\usepackage{amsmath,amssymb,units}
\usepackage{cases}
\usepackage{mathrsfs}
\usepackage{multirow}
\usepackage{caption}
\usepackage{color}
\usepackage{flushend}
\usepackage{float}
\usepackage[ruled,vlined,linesnumbered]{algorithm2e}
\usepackage{algorithmic}
\usepackage{subfigure}
\usepackage{extarrows}
\usepackage{graphicx,graphics,txfonts}
\usepackage{array}
\usepackage{tabularx}
\usepackage{stmaryrd}
\usepackage{diagbox}
\usepackage{balance}
\usepackage{bm}
\usepackage{enumitem}

\vldbTitle{A Sample Proceedings of the VLDB Endowment Paper in LaTeX Format}
\vldbAuthors{Jinfei Liu}
\vldbDOI{https://doi.org/10.14778/xxxxxxx.xxxxxxx}
\vldbVolume{12}
\vldbNumber{xxx}
\vldbYear{2019}

\begin{document}

\newtheorem{definition}{Definition}
\renewcommand{\algorithmicrequire}{\textbf{Requires}}
\newtheorem{theorem}{Theorem}
\newtheorem{lemma}{Lemma}
\newtheorem{assumption}{Assumption}
\newtheorem{axiom}{Axiom}
\newtheorem{example}{Example}
\newtheorem{corollary}{Corollary}
\newtheorem{property}{Property}
\newtheorem{discussion}{Discussion}
\newtheorem{remark}{Remark}

\newcommand{\partitle}[1]{\medskip \noindent \textbf{#1.}}
\newcommand{\subpartitle}[1]{\medskip \emph{#1.}}
\newcommand{\topcaption}{%
\setlength{\abovecaptionskip}{0pt}%
\setlength{\belowcaptionskip}{0pt}%
\caption}

\title{Dealer: End-to-End Data Marketplace with Model-based Pricing }

\maketitle


\begin{abstract}

Data-driven machine learning (ML) has witnessed great successes across a variety of application domains. Since ML model training are crucially relied on a large amount of data, there is a growing demand for high quality data to be collected for ML model training. However, from data owners' perspective, it is risky for them to contribute their data. To incentivize data contribution, it would be ideal that their data would be used under their preset restrictions and they get paid for their data contribution.

In this paper, we take a formal data market perspective and propose the first en\textbf{\underline{D}}-to-\textbf{\underline{e}}nd d\textbf{\underline{a}}ta marketp\textbf{\underline{l}}ace with mod\textbf{\underline{e}}l-based p\textbf{\underline{r}}icing (\emph{Dealer}) towards answering the question: \emph{How can the broker assign value to data owners based on their contribution to the models to incentivize more data contribution, and determine pricing for a series of models for various model buyers to maximize the revenue with arbitrage-free guarantee}. For the former, we introduce a Shapley value-based mechanism to quantify each data owner's value towards all the models trained out of the contributed data. For the latter, we design a pricing mechanism based on models' privacy parameters to maximize the revenue. More importantly, we study how the data owners' data usage restrictions affect market design, which is a striking difference of our approach with the existing methods. Furthermore, we show a concrete realization DP-\emph{Dealer} which provably satisfies the desired formal properties. Extensive experiments show that DP-\emph{Dealer} is efficient and effective.

\end{abstract}

\section{Introduction}\label{sec:Introduction}

Machine learning has witnessed great success across various types of tasks and is being applied in an ever-growing number of industries and businesses. High usability machine learning models depend on a large amount of high-quality training data, which makes it obvious that data are valuable. Recent studies and practices approach the commoditization of data in various ways. A data marketplace sells data either in the direct or indirect (derived) forms. These data marketplaces can be generally categorized based on their pricing mechanisms: 1) data-based pricing, 2) query-based pricing, and 3) model-based pricing.

Data marketplaces with data-based pricing are selling datasets and allow buyers to access the data entries directly, e.g., Dawex \cite{dawex}, Twitter \cite{Twitter}, Bloomberg \cite{Bloomberg}, and Iota \cite{IOTA}. Under these marketplaces, data owners have limited control over their data usage, e.g., privacy abuse, which makes it challenging for the market to incentivize more data owners to contribute. Also, it can be overpriced for buyers to purchase the whole dataset when they are only interested in particular information extracted from the dataset. Therefore, the marketplace operates in an inefficient way that cannot maximize the revenue.

Data marketplaces with query-based pricing \cite{koutris2012query,koutris2013toward}, e.g., Google Bigquery \cite{googlequery}, partially alleviate these shortcomings by charging buyers and compensating data owners on a per-query basis. The marketplace makes decisions about the restrictions on data usage (e.g., return queries with privacy protection \cite{DBLP:journals/tods/LiLMS14}), compensation allocation, and query pricing. However, most queries considered by these marketplaces are too simplistic to support sophisticated data analytics and decision making.


Data marketplaces with model-based pricing \cite{agarwal2019marketplace,DBLP:conf/sigmod/ChenK019,jia2019efficient} have been recently proposed. In \cite{DBLP:conf/sigmod/ChenK019}, the authors focus on pricing a series of model instances depending on their model quality to maximize revenue, while \cite{jia2019efficient} considers how to allocate compensation in a fair way among data owners when their data are utilized for a particular machine learning model of $k$-nearest neighbors ($k$-NN). Thus, they are limited to either end of the marketplace but not both. Most recently, \cite{agarwal2019marketplace} approaches it in a relatively more complete perspective by studying two ends of the marketplace, where strategies for the broker to set the model usage charge from the buyers, and for the broker to distribute compensation to the data owners are proposed. However, \cite{agarwal2019marketplace} oversimplifies the role of the two end entities played in the overall marketplace: the data owners and the model buyers. For example, the data owners still have no means to control the way that their data is used, while the model buyers do not have a choice over the quality of the model that best suits their needs and budgets.

\partitle{Gaps and Challenges}
Though efforts have been made to ensure the broker follows important market design principles in \cite{agarwal2019marketplace,DBLP:conf/sigmod/ChenK019,jia2019efficient}, how the marketplace should respond to the needs of both the data owners and the model buyers is still understudied. It is therefore tempting to ask: how can we build a marketplace dedicated to machine learning models, which can simultaneously satisfy the needs of all three entities, i.e., data owners, broker, and model buyers. We summarize the gaps and challenges from the perspective of each entity as follows.

\begin{itemize}[leftmargin=*]
	\item \emph{Data owners.} Under the existing data marketplace solutions \cite{agarwal2019marketplace,jia2019efficient}, the data owners receive compensation for their data usage allocated by the broker. Except for this, they have no means to set restrictions about their data usage after supplying their data to the broker. The challenge to be addressed is: \emph{How to model the data owners' restrictions and their associated effect on model manufacturing, model pricing, and compensation allocation}?

	\item \emph{Model buyers.} As with the same practice of selling digital commodities in several versions, current work \cite{agarwal2019marketplace,DBLP:conf/sigmod/ChenK019} provides a series of models for sale with different levels of quality. However, their oversimplified noise injection-based version control quantifies the model quality via the magnitude of the noise, which does not directly align with the model buyers' valuation of the model in terms of its utility. \emph{How do we incorporate the model buyers' perspective in the model valuation and use it to optimize both model pricing and model manufacturing}?

	\item \emph{Broker.} The data owners' restrictions and the model buyers' estimation of the model value should be taken into consideration by the broker when making market decisions, e.g., compensation allocation and model pricing. \emph{How can the broker align the two ends' requirements with the already complicated market design principals in an efficient and effective way}? That is, how can the broker remain competitive (e.g., train higher utility model with the data owner restrictions), while maximizing its revenue to maintain a sustainable data market?
\end{itemize}

\vspace{-1em}
\begin{figure}[htb]
 \centering
 \includegraphics[width=0.5\textwidth]{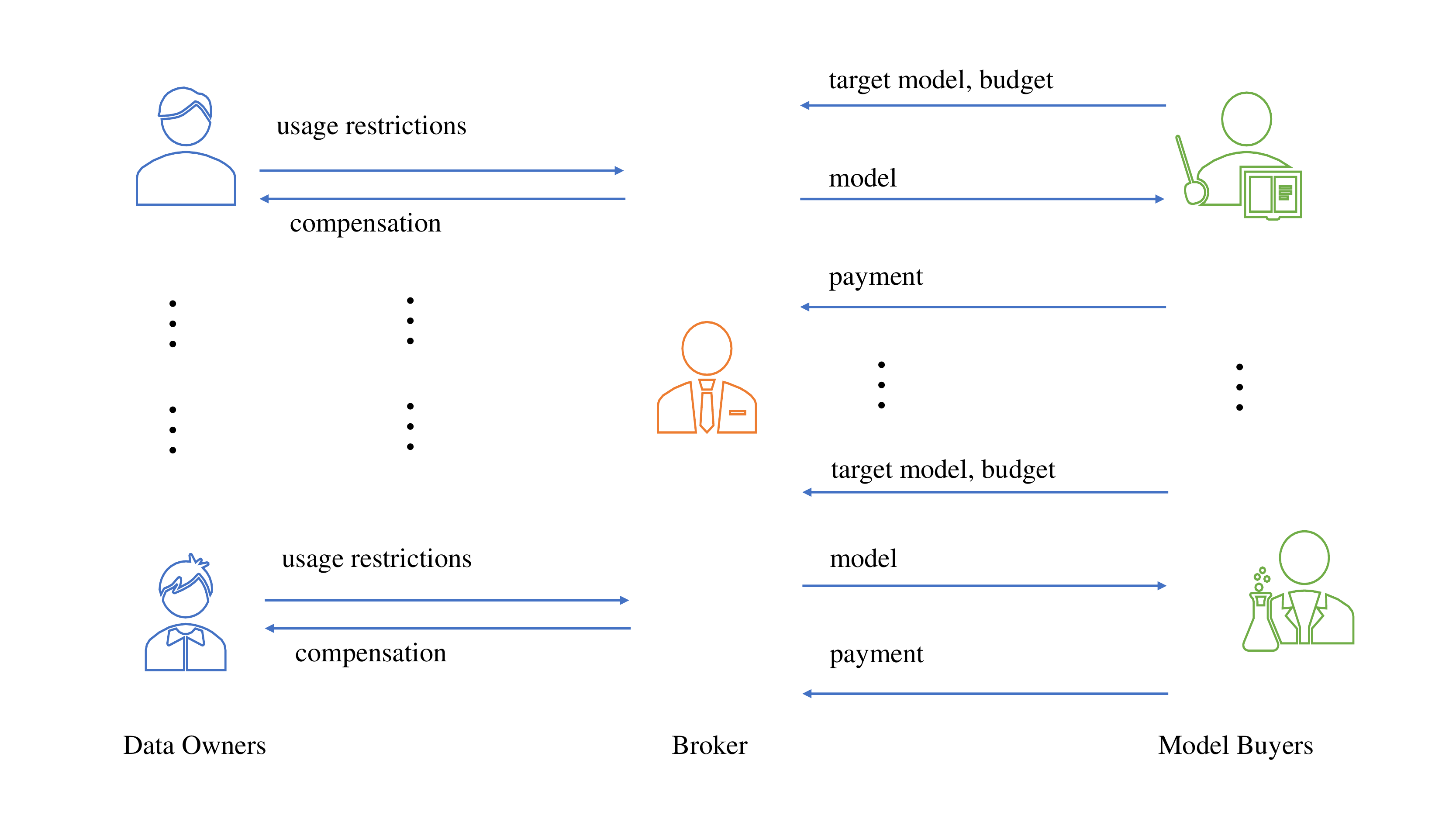}
 \vspace{-1em}
 \caption{An end-to-end data marketplace with model-based pricing framework.}
 \label{fig:framework}
\end{figure}
\vspace{-1em}

\partitle{Contributions}
In this paper, we bridge the above gaps and address the identified challenges by proposing \emph{Dealer}: an en\textbf{\underline{D}}-to-\textbf{\underline{e}}nd d\textbf{\underline{a}}ta marketp\textbf{\underline{l}}ace with mod\textbf{\underline{e}}l-based p\textbf{\underline{r}}icing. \emph{Dealer} provides both abstract and general mathematical formalization for the end-to-end marketplace dynamics (Section \ref{sec:FrameWorkDesiderata}) and a specific and practical differentially private marketplace with algorithm designs that ensures differential privacy of the data owners, one type of instantiation of the model restriction.

\emph{End-to-End Mathematical Formalization of the Marketplace.}
We first propose an abstract mathematical formalization for \emph{Dealer} with an emphasis on the understudied parts of the marketplace as discussed above. An illustration is provided in Figure \ref{fig:framework}, which includes three main entities and their interactions. In this general \emph{Dealer} (Gen-\emph{Dealer}), we define the abilities and constraints for the three entities (i.e., data owners, broker, and model buyers). From the data owners' perspective, Gen-\emph{Dealer} aims to 1) allow the data owners to restrict their data usage by the broker; and 2) have the option to receive extra compensation if they are willing to partially relax their restrictions. From the model buyers' perspective, Gen-\emph{Dealer} provides a utility measure suitable for the model buyers' model valuation, based on which the potential model buyers provide their willingness to purchase and payment estimation. From the broker's perspective, Gen-\emph{Dealer} depicts the full marketplace dynamics through abstract behavior functions. In addition to the new features novely brought into the market design consideration by Gen-\emph{Dealer}, commonly recognized market design principals are also well-accommodated into our approach. Prominent ones include: 1) compensation allocation in a fair and rational way, e.g., based on the widely-accepted Shapley value notion; 2) arbitrage-free in model pricing, which prevents the model buyers from taking advantage of the marketplace by combining lower-tier models sold for cheaper prices into a higher-tier model to escape the designated price for that tier.

\emph{Marketplace Instantiation with Differential Privacy.} We also provide a concrete end-to-end differentially private data marketplace instance. We focus on empirical risk minimization, a widely used and well-studied family of supervised machine learning models. For the data owners' restrictions, we consider the privacy restriction, which is arguably the most concerned issue for personal data contributors. Differential privacy (DP) \cite{dwork2006calibrating,dwork2014algorithmic}, the de facto standard in privacy-preserving data analysis nowadays, is introduced to exemplify the data owners' restriction requirements and we will refer the differentially private data marketplace with model-based pricing instance by DP-\emph{Dealer}. In DP-\emph{Dealer}, the marketplace sells a series of differentially private models to respect data owners' privacy restrictions. The higher tier models correspond to models trained on data subsets contributed by the data owners with lower DP restrictions. On the contrary, the lower tier models are trained with higher DP restriction data subsets. We will consider two types of data owner DP restrictions: 1) hard restriction which has a rigid cutting point beyond that the data owners' data cannot be used for training; and 2) negotiable restriction which has a negotiable range within that the marketplace still has the option to use the data but with extra compensation. At the model buyers' end, DP-\emph{Dealer} addresses the challenge of mismatched model tier ranking standard by converting the model ``manufacturing'' tier ranking standard to the model utility standard adopted by the model buyers in making purchasing decisions. DP-\emph{Dealer} accommodates market design principals like fair compensation allocation and arbitrage-free model pricing. To support the end-to-end marketplace dynamics with all design considerations, DP-\emph{Dealer} establishes constrained objective functions and develops efficient algorithms to optimize market decisions.

We briefly summarize our contributions as follows.
\begin{itemize}[leftmargin=*]
       \item A general end-to-end data marketplace with model-based pricing framework Gen-\emph{Dealer} which is the first systematic study that includes all market participants (i.e., data owners, broker, and model buyers). Gen-\emph{Dealer} formalizes the abilities and restrictions of the three entities and models the interactions among them.
       \vspace{-0.5em}

       \item A differentially private data marketplace with model-based pricing DP-\emph{Dealer} which instantiates the general framework. In addition to incorporating market design principals, DP-\emph{Dealer} proposes two data owner restriction schemes and provides the utility estimation for the model buyers to choose models best suiting their needs. DP-\emph{Dealer} formulates a series of optimization problems and develops efficient algorithms to make the market decisions.
       \vspace{-0.5em}

       \item A series of experiments are conducted to justify the design of DP-\emph{Dealer} and verify the efficiency and effectiveness of the proposed algorithms.
\end{itemize}

\partitle{Organization} The rest of the paper is organized as follows. Section \ref{sec:Related} presents the related work. Section \ref{sec:Definition} provides the background information, including the concept of Shapley value and its computation, the machine learning model exemplified in this paper, and differential privacy related definitions and properties. We provide the first end-to-end data marketplace with model-based pricing formalization and discuss the desiderata in Section \ref{sec:FrameWorkDesiderata}. A concrete instance of differentially private data marketplace with efficient algorithms is derived in Section \ref{sec:FrameWorkInstance}. We report the experimental results and findings in Section \ref{sec:Experiments}. Finally, Section \ref{sec:Conclusion} draws a conclusion and discusses future work.

\section{Related Work}\label{sec:Related}
In this section, we discuss related work on data pricing and compensation allocation.

\subsection{Data Pricing}

Ghosh et al. \cite{DBLP:conf/sigecom/GhoshR11} initiated the study of markets for private data using differential privacy. They modeled the first framework in which data buyers would like to buy sensitive information to estimate a population statistic. They defined a property named envy-free for the first time. Envy-free ensures that no individual would prefer to switching their payment and privacy cost with each other. Guruswami et al. \cite{DBLP:conf/soda/GuruswamiHKKKM05} studied the optimization problem of revenue maximization with envy-free guarantee. They investigated two cases of inputs: unit demand consumers and single minded consumers, and showed the optimization problem is APX-hard for both cases, which can be efficiently solved by a logarithmic approximation algorithm. Li et al. \cite{DBLP:conf/icdt/LiLMS13, DBLP:journals/tods/LiLMS14, DBLP:journals/cacm/LiLMS17} presented the first theoretical framework for assigning value to noisy query answers as function of their accuracy, and for dividing the price among data owners who deserve compensation for their loss of privacy. They defined an enhanced edition of envy-free, which is named arbitrage-free. Arbitrage-free ensures the data buyer cannot purchase the desired information at a lower price by combing two low-price queries.

Lin et al. \cite{DBLP:journals/pvldb/LinK14} proposed necessary conditions for avoiding arbitrage and provide new arbitrage-free pricing functions. They also presented a couple of negative results related to the tension between flexible pricing and arbitrage-free, and illustrated how this tension often results in unreasonable prices. In addition to arbitrage-free, Koutris et al. \cite{DBLP:journals/jacm/KoutrisUBHS15} proposed another desirable property for the pricing function, discount-free, which requires that the prices offer no additional discounts than the ones specified by the broker. In fact, discount-free is the discrete version of arbitrage-free. Furthermore, they presented a polynomial time algorithm for pricing generalized chain queries.

Recently, Chawla et al. \cite{DBLP:journals/pvldb/ChawlaDKT19} investigated three types of succinct pricing functions and studied the corresponding revenue maximization problems. Due to the increasing pervasiveness of machine learning based analytic, there is an emerging interest in studying the cost of acquiring data for machine learning. Chen et al. \cite{DBLP:conf/sigmod/ChenK019} proposed the first and the only existing model-based pricing framework in which instead of pricing the data, directly prices machine learning model instances. They formulated an optimization problem to find the arbitrage-free price that maximizes the revenue of the broker, and proved such optimization problem is coNP-hard. However, their work only focuses on the interactions between the broker and the model buyers. Furthermore, they assume there is only one survey price for each model, which is too simplified.

\subsection{Compensation Allocation}
An acquiescent method to evaluate data importance/value to a model is leave-one-out (LOO) which compares the difference between the predictor's performance when trained on the entire dataset and the predictor's performance when trained on the entire dataset minus one point \cite{DBLP:journals/pr/CawleyT03LOO}. However, LOO does not satisfy all the ideal properties that we expect for the data valuation. For example, given a point $p$ in a dataset, if there is an exact copy $p'$ in the dataset, removing $p$ from this datasets does not change the predictor at all since $p'$ is still there. Therefore, LOO will assign zero value to $p$ regardless of how important $p$ is.

Shapley value is a concept in cooperative game theory, which was named in honor of Lloyd Shapley \cite{shapley1953value}. Shapley value is the only value division scheme used for compensation allocation that meets three desirable criteria, group rationality, fairness, and additivity \cite{jia2019efficient}. Combining with its flexibility to support different utility functions, Shapley value has been extensively employed in the data pricing field \cite{agarwal2019marketplace, DBLP:conf/icml/AnconaOG19, DBLP:conf/icml/GhorbaniZ19, jia2019efficient}. One major challenge of applying Shapley value is its prohibitively high computational complexity. Evaluating the exact Shapley value involves the computation of the marginal utility of each user to every coalition, which is $\sharp P$-complete \cite{DBLP:journals/ai/FatimaWJ08}. Such exponential computation is clearly impractical for evaluating a large number of training points. Even worse, for machine learning tasks, evaluating the utility function is extremely expensive as machine learning tasks need to train models. The worst case is that we need to train $O(2^n)$ models for computing the exact Shapley value for each data owner.

A number of approximation methods have been developed to overcome the computational hardness of finding the exact Shapley value. The most representative method is Monte Carlo method \cite{DBLP:journals/cor/CastroGT09,DBLP:journals/ai/FatimaWJ08}, which is based on the random sampling of permutations. However, the time cost is still prohibitively high due to the high training cost of deep learning models. Therefore, Ghorbani et al. \cite{DBLP:conf/icml/GhorbaniZ19} and Ancona et al. \cite{DBLP:conf/icml/AnconaOG19} illustrated how to compute the approximate Shapley value by performing stochastic gradient descent on one data point at a time.

\section{Background and Preliminaries}\label{sec:Definition}
In this section, we introduce the background and preliminaries of \emph{Dealer}. We summarize the frequently used notations in Table \ref{tab:notations}. In particular, we denote the data owners by $\mathcal{O}_1,...,\mathcal{O}_i,...,\mathcal{O}_n$, a series of models prepared by the broker for sale by $\mathcal{M}^1,...,\mathcal{M}^m,...,\mathcal{M}^M$, and the model buyers by $\mathcal{B}_1,...\mathcal{B}_{k},...,\mathcal{B}_{K}$.

\begin{table}[htb]\centering
\caption{The summary of notations.}\label{tab:notations}
\vspace{-1em}
{%
\footnotesize
\begin{tabular}{|c|c|}
\hline
Notation & Definition\\
\hline
$\mathcal{O}_i$ & the $i^{th}$ data owner\\
\hline
$\mathcal{M}^m$ & the $m^{th}$ model\\
\hline
$\mathcal{B}_k$ & the $k^{th}$ model buyer\\
\hline
$\bm{Z}_{train}=\{\bm{z}_1,\bm{z}_2,...,\bm{z}_n\}$ & training dataset\\
\hline
$\bm{X}_{train}=\{\bm{x}_1,\bm{x}_2,...,\bm{x}_n\}$ & features of training dataset\\
\hline
$\bm{y}_{train}=\{y_1,y_2,...,y_n\}$ & labels of training dataset\\
\hline
$\bm{z}_i=\{\bm{x}_i,y_i\}$ & the $i^{th}$ training data\\
\hline
$\bm{Z}_{test}$ & testing dataset\\
\hline
$\bm{X}_{test}$ & features of testing dataset\\
\hline
$\bm{y}_{test}$ & labels of testing dataset\\
\hline
$\mathcal{U}$ & model utility\\
\hline
$\mathcal{SV}$ & Shapley value\\
\hline
$\mathcal{UV}$ & utility valuation\\
\hline
$\mathcal{UF}$ & utility function\\
\hline
$\mathcal{MR}$ & model risk factor\\
\hline
$\mathcal{MB}$ & manufacturing budget\\
\hline
$\mathcal{DR}$ & data owner restriction function\\
\hline
$\epsilon,\delta$ & parameters for the DP\\
\hline
$\bm{bc}$ & base compensation\\
\hline
$\bm{ec}$ & extra compensation\\
\hline
$\bm{tm}$ & target model\\
\hline
$\langle p(\epsilon^1),p(\epsilon^2),...,p(\epsilon^M)\rangle$ & optimal pricing\\
\hline
$(m,sp^m[j])$ & survey price point\\
\hline
$(m,p^m[j])$ & complete price point\\
\hline
\end{tabular}}
\end{table}%

\subsection{Established Market Design Principals}
There are a number of market design principals, which are also considered in data marketplaces with data-based pricing, query-based pricing, and model-based pricing. In the following, we review the most widely adopted ones and introduce related techniques.

\subsubsection{Fairness in Compensation Allocation: Shapley Value}\label{sub:ApprShapey}

Shapley value based compensation is a prevalently adopted approach mostly due to its theoretical properties, especially the fairness. Shapley value measures the marginal improvement of model utility contributed by $\bm{z}_i$ of data owner $\mathcal{O}_i$, averaged over all possible coalitions of the data owners. The formal Shapley value definition of data owner $\mathcal{O}_i$ is shown as follows.
\begin{equation}\label{equ:shapleyValue}
  \mathcal{SV}_i=\sum_{\bm{S}\subseteq \{\bm{z}_1,...,\bm{z}_n\}\setminus \bm{z}_i}\frac{\mathcal{U}(\bm{S}\cup \{\bm{z}_i\})-\mathcal{U}(\bm{S})}{\binom{n-1}{|\bm{S}|}}
\end{equation}
where $\mathcal{U}(\cdot)$ is the utility of the model trained by a coalition of the data owners, and the model utility is tested on the training dataset.


\partitle{Monte Carlo Simulation Method}
Since the exact Shapley value computation is based on enumeration which is prohibitively expensive, we adopt a commonly used Monte Carlo simulation method \cite{DBLP:journals/cor/CastroGT09,DBLP:journals/ai/FatimaWJ08} to compute the approximate Shapley value. We first sample random permutations of the data points corresponding to different data owners, and then scan the permutation from the first element to the last element and calculate the marginal contribution of every new data point. Repeating the same procedure over multiple Monte Carlo permutations, the final estimation of the Shapley value is simply the average of all the calculated marginal contributions. This Monte Carlo sampling gives an unbiased estimate of the Shapley value. In practical applications, we generate Monte Carlo estimates until the average has empirically converged and the experiments show that the estimates converge very quickly. Therefore, Monte Carlo simulation method can control the degree of approximation, i.e., the more permutations, the better the accuracy . The detailed algorithm is shown in Algorithm \ref{Alg:MCShapley}, where $|\pi|$ is the number of permutations. The larger the $|\pi|$, the more accurate the computed Shapley value.

\begin{algorithm}[thb] \caption{Monte Carlo Shapley value computation.}\label{Alg:MCShapley}
\SetKwInOut{Input}{input}\SetKwInOut{Output}{output}

\Input{$\bm{Z}_{train}=(\bm{X}_{train},\bm{y}_{train})$ and $\bm{Z}_{test}=(\bm{X}_{test},\bm{y}_{test})$.}
\Output{Shapley value $\mathcal{SV}_i$ for each data $\bm{z}_i$ in $\bm{Z}_{train}$.}

initialize $\mathcal{SV}_i=0$\;

    \For{k=1 to $|\pi|$}{
        we have a training dataset ordered in $\pi^k$, $\bm{Z}_{train}^k=\{\bm{z}_{\pi^k_1},\bm{z}_{\pi^k_2},...,\bm{z}_{\pi^k_n}\}$\;
        \For{i=1 to n}{
            $\mathcal{SV}(\bm{z}_{\pi^k_i})=\mathcal{U}(\{\bm{z}_{\pi^k_1},...,\bm{z}_{\pi^k_i}\})-\mathcal{U}(\{\bm{z}_{\pi^k_1},...,\bm{z}_{\pi^k_{i-1}}\})$\;
            $\mathcal{SV}_{\pi^k_i}=\mathcal{SV}_{\pi^k_i}+\mathcal{SV}(\bm{z}_{\pi^k_i})$\;
        }
    }
\end{algorithm}

\subsubsection{Versioning}
For a perfect world, the broker would sell a personalized model to each model buyer at a different price to maximize the revenue. However, such personalized pricing is rarely possible in practical applications. On the one hand, it is expensive to train different models for different model buyers. On the other hand, it is difficult to set an array of prices for the same model. Even if we could, it would be impossible to get model buyers to stay within their intended pricing strata rather than to look for the lowest price. Finally, the broker runs the risk of annoying or even alienating model buyers if they charge different prices for the same model.

There is a practical way to set different prices for the same training dataset without incurring high costs or offending model buyers. We can do it by offering the training dataset in different versions designed to attract different types of model buyers. With this strategy, which is called versioning \cite{shapiro1998versioning}, model buyers segment themselves. The model version they choose reveals the value they place on the training dataset and the price they are willing to pay. Therefore, in our model marketplace, the broker would train $M$ different model versions by injecting different noise to different subsets of training data contributed by the data owners.

\subsubsection{Arbitrage-free in Model Pricing}\label{sub:AF}
Arbitrage is possible when the ``better'' model can be obtained more cheaply than the advertised price by combining two or more worse models with lower price. Arbitrage complicates the interactions between the broker and the model buyers, i.e., the model buyers need to carefully choose the models to achieve the lowest price, while the broker may not achieve the revenue intended by her advertised prices. Therefore, an arbitrage-free pricing function is highly desirable. We say a pricing function is arbitrage-free if it satisfies the following two properties proved in \cite{DBLP:conf/sigmod/ChenK019, DBLP:journals/cacm/LiLMS17}.

\begin{property}\partitle{(Monotonicity)}
Given a function $f: (\mathbb{R}^+)^k \rightarrow \mathbb{R}^+$, we say $f$ is monotone if and only if for any two vectors $\mathbf{x},\mathbf{y}\in (\mathbb{R^+})^k$, $\mathbf{x}\leq \mathbf{y}$, we have $f(\mathbf{x})\leq f(\mathbf{y})$.
\end{property}


\begin{property}\partitle{(Subadditivity)}
Given a function $f: (\mathbb{R}^+)^k \rightarrow \mathbb{R}^+$, we say $f$ is subadditive if and only if for any two vectors $\mathbf{x},\mathbf{y}\in (\mathbb{R^+})^k$, we have $f(\mathbf{x}+\mathbf{y})\leq f(\mathbf{x})+f(\mathbf{y})$.
\end{property}


\subsection{Machine Learning Models}
In this paper, we focus on the Empirical Risk Minimization (ERM), which is a widely-applied tool in machine learning. Denote the training dataset by $\bm{Z}_{train}:= \{\bm{z}_i\},\ i={1,2,...,n}$, where $\bm{z}_i \sim \mathcal{D}$ and $\bm{z}_i=(\bm{x}_i,y_i)$. The $\bm{x}_i\in \mathbb{R}^d$ is the $d$-dimensional feature vector and $y_i$ is the response value which can be $\{-1,+1\}$ for binary classification task or $[0,1]$ for regression task. The ERM has the following objective function:
\begin{equation}
    \arg\min_{\bm{w}\in\Omega} \bm{L}(\bm{w};\bm{Z}_{train})+\lambda\|\bm{w}\|_2^2 = \arg\min_{\bm{w}\in\Omega} \sum_{i=1}^{n} \frac{1}{n} l(\bm{w};\bm{z}_i)+\lambda\|\bm{w}\|_2^2,
\end{equation}
where $\bm{L}(\bm{w};\bm{Z}_{train}) = \sum_{i=1}^{n}l(\bm{w};\bm{z}_i)$ is the empirical loss averaged from losses taken from all $\bm{z}_i$, $\lambda\|\bm{w}\|_2^2$ is the regularizer and $\Omega$ is the constraint set. In this paper, we focus on models with convex, Lipschitz continuous, and smooth loss (with respect to $\bm{w}$) functions. The formal definitions are in the following.
\begin{definition} (Convex Loss Function)
A loss function $l(\bm{w}):\ \mathbb{R}^d\to \mathbb{R}$ is called convex if for all $\bm{w}_1,\bm{w}_2\in\mathbb{R}^d$,
\begin{equation}
    |l(\bm{w}_1) - l(\bm{w})_2| \geq \langle \nabla l(\bm{w}_2), \bm{w}_1-\bm{w}_2\rangle.
\end{equation}
In addition, if $l(\bm{w}_1) - l(\bm{w})_2 \geq \langle \nabla l(\bm{w}_2), \bm{w}_1-\bm{w}_2\rangle + \frac{\mu}{2}\|\bm{w}_1-\bm{w}_2\|_2^2$, for $\mu> 0$, $l(\bm{w})$ is $\mu$-strongly convex.
\end{definition}

\begin{definition} (Lipschitz continuous Loss Function)
A loss function $l(\bm{w}):\ \mathbb{R}^d\to \mathbb{R}$ is called $L$-Lipschitz continuous if for all $\bm{w}_1,\bm{w}_2\in\mathbb{R}^d$,
\begin{equation}
    |l(\bm{w}_1) - l(\bm{w})_2| \leq L\|\bm{w}_1-\bm{w}_2\|_2.
\end{equation}
\end{definition}

\begin{definition} (Smooth Loss Function)
A loss function $l(\bm{w}):\ \mathbb{R}^d\to \mathbb{R}$ is called $\beta$-smooth if for all $\bm{w}_1,\bm{w}_2\in\mathbb{R}^d$,
\begin{equation}
    l(\bm{w}_1) - l(\bm{w})_2 \leq \langle \nabla l(\bm{w}_2), \bm{w}_1-\bm{w}_2\rangle + \frac{\beta}{2}\|\bm{w}_1-\bm{w}_2\|_2^2
\end{equation}
\end{definition}

We focus on the loss functions satisfying the above assumptions like least square loss, logistic loss, and smoothed hinge loss. In machine learning, these losses are thoroughly studied with theoretical properties like generalization performance, and easy to use with efficient optimization algorithm with guaranteed convergence. Furthermore, their differentially private versions are also equipped with efficiency, utility, and privacy guarantees. As a result, we focus on this particular type of machine learning models in our market design for a thorough understanding of our proposed market. However, we would like to mention that most of our algorithms can be applied to other types of machine learning models like the popular deep learning family.

\subsection{Differential Privacy}
Differential privacy is a formal mathematical tool for rigorously providing privacy protection. However, none of the existing (published) data marketplace with model-based pricing paper has considered it and it is still unknown how it can be incorporated in data marketplace with model-based pricing and how it affects market designs when adopted.

\begin{definition}(Differential Privacy)
A randomized algorithm $\mathcal{A}$ is $(\epsilon, \delta)$-differentially private, if for any pair of datasets $\bm{S}$ and $\bm{S}'$ that differs in one data sample, and for all possible output $\mathcal{O}$ of $\mathcal{A}$, the following holds,
\begin{equation}
    \mathbb{P}[\mathcal{A}(\bm{S})\in\mathcal{O}] \leq e^{\epsilon}\mathbb{P}[\mathcal{A}(\bm{S}')\in\mathcal{O}] + \delta,
\end{equation}
where the probability is taken over the randomness of $\mathcal{A}$.
\end{definition}
In practice, for a meaningful DP guarantee, the parameters are chosen as $0 < \epsilon \leq 1$, $\delta \ll \frac{1}{N}$, where $N$ is the number of data samples.

\begin{lemma} (Simple Composition)\label{lem:simpleCompo}
Let $\mathcal{A}_j$ be an $(\epsilon_j,\delta_j)$-differentially private algorithm. We have $\mathcal{A} = (\mathcal{A}_1,...,\mathcal{A}_J)$ is $(\sum_{j=1}^J\epsilon_j,\sum_{j=1}^J\delta)$-differentially private.
\end{lemma}
Lemma \ref{lem:simpleCompo} is essential for DP mechanism design and analysis, which enables algorithm designers to compose elementary DP operations into a more sophisticated one. More importantly, we will show that it plays a crucial role in model market design as well. That is, Lemma \ref{lem:simpleCompo} determines that DP is an appropriate mechanism for versioning models, based on which prices have to satisfy the arbitrage-free property.

\begin{definition} ($\ell _2$-sensitivity)
A function $f:\ \mathcal{D}^{N}\to
\mathbb{R}^d$ has $\ell_2$ sensitivity $\Delta _2$ if
\begin{equation}
    \max _{\text{neighboring }\bm{S},\bm{S}'} \|f(\bm{S})-f(\bm{S}')\|_2 = \Delta _2
\end{equation}
\end{definition}
For training ERM with DP restriction, a popular method is the objective perturbation, which perturbates the objective function of the model with quantified noise. Conventional objective perturbation only supports DP guarantee for the exact optimum, which is hardly achievable in practice. In this paper, we follow the enhanced objective perturbation called \emph{approximate minima perturbation} \cite{bassily2019private, iyengar2019towards}, which allows solving the perturbed objective up to $\alpha$ approximation. It uses a two-phase noise injection strategy that perturbates both the objective and the approximate output. The detail is summarized in Algorithm \ref{Alg:DPtraining}. The algorithm trains a $(\epsilon,\delta)$-DP model based on training dataset $\bm{Z}_{train}$, which outputs model parameter $\bm{w}_{DP}$. In particular, Line 2 perturbates the model with calibrated noise $\bm{N}_1$, Line 3 optimizes the perturbated model which is followed by an output perturbation with noise $\bm{N}_2$. Finally, $\bm{w}_{DP}$ is obtained by a projection to the constrained set $\Omega$.

\begin{algorithm}[thb] \caption{Objective perturbation for differentially private ERM training.}\label{Alg:DPtraining}
\SetKwInOut{Input}{input}\SetKwInOut{Output}{output}

\Input{$\bm{Z}_{train}$ and $(\epsilon,\delta)$.}
\Output{$\bm{w}_{DP}$.}

Sample $\bm{N}_1 \sim \mathcal{N}(\bm{0}_d,\sigma _1^2 \bm{I}_d)$, where $\sigma_1=\frac{20L^2 \log(1/\delta)}{\epsilon ^2}$\;

Objective Perturbation: $\mathcal{L}_{OP}(\bm{w}) = \mathcal{L}(\bm{w};\bm{Z}_{train}) +\lambda\|\bm{w}\|_2^2 +\frac{1}{n}\langle \bm{N}_1,\bm{w}\rangle$\;

Optimize $\mathcal{L}_{OP}(\bm{w})$ to obtain $\alpha$ approximate solution $\hat{\bm{w}}$\;

Sample $\bm{N}_2 \sim \mathcal{N}(\bm{0}_d,\sigma _2^2 \bm{I}_d)$, where $\sigma_2=\frac{40\alpha \log(1/\delta)}{\lambda\epsilon ^2}$\;

\Return $\bm{w}_{DP} = proj_{\mathcal{W}}(\hat{\bm{w}}+\bm{N}_2)$\;
\end{algorithm}

\section{General End-to-End Data Marketplace with Model-based Pricing}\label{sec:FrameWorkDesiderata}

In this section, we propose Gen-\emph{Dealer} to model the entire end-to-end data marketplace with model-based pricing. Gen-\emph{Dealer} models the following three aspects of the data marketplace: 1) the functionalities and restrictions of the three participating entities; 2) interactions between the participating entities; 3) market decisions taken by the broker. In addition to the well recognized market design principals considered by previous work, we bring new considerations for the two ends of the marketplace (i.e., the data owners and the model buyers), and study how the marketplace interactions should respond to their requirements and how the market decisions will be affected to maintain a sustainable (or even profitable) data marketplace with model-based pricing.

\subsection{Formalization of Marketplace Entities }\label{sub:marketPartForm}
To begin with, we formalize the three participating entities as follows.

\partitle{Data Owners}
Data owners can be professional institutes, organizations, or individuals. In this paper, we focus on the individual case where each data owner $\mathcal{O}_i$ contributes their own data $\bm{z}_i$. For organizations or institutes which collects multiple individual owners' data for sale, we treat every atomic data tuple within the dataset as an individual data owner. Individual data owners are interested in contributing their data for compensation, e.g., coupons, exclusive sale, cashback, but are cautious about their data usage, e.g., personal privacy exposure. The three main functionalities of the data owners are shown as follows.
\begin{enumerate}[leftmargin=*]
       \item Contributing Data: each data owner can contribute the personal data to the broker. Denote the data of data owner $\mathcal{O}_i$ by $\bm{z}_i = (\bm{x}_i,y_i)$, where $\bm{x}_i$ is the feature vector and $y_i$ is the response vector.

       \item Setting Usage Restriction: the data owners will set restrictions on how their data can be used, e.g., what types of models or how many models their data can be used. In this paper, we consider a natural strategy, where $\mathcal{O}_i$ sets one restriction per each model to train. That is, for model $\mathcal{M}^m$, $\mathcal{O}_i$ provides \textbf{D}ata \textbf{R}estriction function $\mathcal{DR}_i^m$. Two types of restriction functions are modeled in this paper. The first is the simpler ``in or out'' restriction, where $\mathcal{DR}_i^m = \{0,1\}$ is an indication function providing hard restriction on whether $\bm{z}_i$ is allowed to be used for training the $m^{th}$ model ($\mathcal{DR}_i^m = 1$) or not ($\mathcal{DR}_i^m = 0$). The second provides negotiable ranges that $\mathcal{DR}_i^m$ is not only a function of the model tier but also related to \textbf{e}xtra \textbf{c}ompensation ($\bm{ec}_i^m$): $\mathcal{DR}_i^m(\bm{ec}_i^m) = \{0,1\}$. That is, $\bm{z}_i$ can be used for training model $\mathcal{M}^m$ if the extra compensation $\bm{ec}_i^m$ satisfies $\mathcal{O}_i$'s expectation.

       \item Receiving Compensation: After the sales, the data owners receive compensation based on their data usage. In addition, extra compensation will be paid if their data is used after negotiation. In particular, the extra compensation $\bm{ec}_i^m$ is a function of risk factor $risk^m$ to be introduced in the next subsection: $\bm{ec}_i^m = 0$, if $risk_{i}^m \leq\mathcal{MR}^m$ and is an nondecreasing nonnegative function of $risk_i^m - \mathcal{MR}^m$ if $risk_{i}^m > \mathcal{MR}^m$.
\end{enumerate}

\begin{discussion}
       Compared to existing machine learning model marketplace designs which limit the data owners actions to merely contribute data and receive compensation, our formulation has the following two strengths: 1) it allows the data owners to set data usage restrictions; 2) the negotiable-type restrictions help data owner to better estimate the value of their personal data, which is often difficult for individuals who have limited market information and evaluation of data usage risk. We believe that by better modeling the data owner by providing the rights on setting data usage restrictions and receiving extra compensation, it will eventually incentivize more data owners to contribute their data.
\end{discussion}

\begin{assumption}
The data owners do not fake data. Each data owner only contributes one data sample and all data are independent.
\end{assumption}

\begin{remark}
In practice, each data owner may have multiple atomic data tuples, and those data tuples among different data owners may be correlated. Therefore, we need to consider their relationship when we allocate compensation. For the correlated data, we leave it as an open question.
\end{remark}

\partitle{Model Buyers}
Model buyers can be industries or everyday users, who are interested in purchasing machine learning models to either integrate into their product or support certain decision making. They have very different budgets and model utility requirements. In this paper, we focus on the single minded model buyers who will purchase at most one model. There are two functionalities of the model buyers as follows.
\begin{enumerate}[leftmargin=*]
       \item Providing Purchase Willingness: $\mathcal{B}_k$ provides the purchase willingness by providing $(\bm{tm}_k,v_k)$, where the \textbf{t}arget \textbf{m}odel $\bm{tm}_k\in\{1,...,M\}$ indicates target model of $\mathcal{B}_k$ and $v_k$ is the purchasing budget. We note that the number of potential model buyers can be very large, but the purchase willingness of hundreds of sampled model buyers is enough for the broker to make market decision.

       \item Model Transaction: $\mathcal{B}_k$ decides to purchase the target model or not by comparing the released price $p(\bm{tm}_k)$ by the broker and her budget.
\end{enumerate}

\partitle{Broker}
The broker collects data from the data owners and trains a series of machine learning models for sale to the model buyers. Let the \textbf{M}odels be $\mathcal{M}^1, ..., \mathcal{M}^m, ..., \mathcal{M}^M$. As discussed, a common versioning strategy is to sell the models in various tiers, from the lowest tier (say $\mathcal{M}^1$) to the highest tier (say $\mathcal{M}^M$). For these models, it sets prices $\langle p(\mathcal{M}^1),...,p(\mathcal{M}^M)\rangle$. In addition, we assume the broker is honest in the sense that it will strictly follow contract with the data owners, e.g., respecting their usage restrictions, allocating the compensation based on the true data usage. The broker also wishes to remain competitive by training the best model for each tier within that tier's resource budget. More importantly, the broker interacts with both the data owners and the model buyers, and makes various market decisions, which are detailed in the later sections. Finally, we assume there is a model risk factor $\mathcal{MR}^m$ associated with each model $\mathcal{M}^m$ to be detailed in Section \ref{subsec:marketDecisionMak}. The model risk measures how large certain risk is to the data owner who participates in the model training. When the risk is low, the broker usually puts high restrictions on data usage, which often leads to limited data information extraction. As a result, the lower risk model often corresponds to lower tier models coming with a lower price. We denote such connection by $\langle p(\mathcal{MR}^1),...,p(\mathcal{MR}^M)\rangle$, where $\mathcal{MR}^1 \leq ... \leq \mathcal{MR}^M$ and the price satisfies arbitrage-free with respect to the model risk. In a more generalized market setting, multiple brokers can co-exist, which forms a competitive relationship. To focus on the principal functionalities of a broker, we follow existing work \cite{agarwal2019marketplace,DBLP:conf/sigmod/ChenK019,jia2019efficient} to consider only a single broker case in this paper.

\subsection{Formalization of Marketplace Dynamics}
In this subsection, we formalize the data marketplace dynamics, which consist of the interaction between the data owner and the broker, the interaction between the model buyer and the broker.

\partitle{Interaction between Data Owner and Broker}

\begin{itemize}[leftmargin=*]
       \item Data Collection: The broker posts model tiers $\mathcal{M}^1,...,\mathcal{M}^M$ and explains each tier's model risk $\mathcal{MR}^m$ to the data owners and how each tier will possibly be compensated if a data owner chooses to participate. The data owner $\mathcal{O}_i$ contributes data $\bm{z}_i$ and set data usage restriction $\mathcal{DR}_i^m$ as well as a potential extra compensation $\bm{ec}_i^m$ if he is willing to negotiate, for all $m=1,...,M$. Recall that $\mathcal{DR}_i^m$ and $\bm{ec}_i^m$ are functions of $\mathcal{MR}^m$.

       \item Compensation Allocation: After training the models under data usage restrictions (see \emph{Model Training} part in the next subsection), the broker pays compensation to the data owners according to the \emph{Compensation Allocation} market decision algorithm detailed in the next subsection. Three key quantities are \textbf{U}tility \textbf{V}aluation to model $\mathcal{M}^m$: $\mathcal{UV}_i^m$, its \textbf{b}ase \textbf{c}ompensation: $\bm{bc}_i^m$, and \textbf{e}xtra \textbf{c}ompensation: $\bm{ec}_i^m$.
\end{itemize}

\partitle{Interaction between Model Buyer and Broker}

\begin{itemize}[leftmargin=*]
       \item Market Survey: The broker posts the models tiers $\mathcal{M}^1,...,\mathcal{M}^M$ to the potential model buyers. Each model buyer $\mathcal{B}_k$ provides purchase willingness $(\bm{tm}_k,v_k)$.
       \item Model Transaction: The broker makes the model pricing $\langle p^1,...,p^M\rangle$ (see \emph{Model Pricing} in the next subsection) based on $(\bm{tm}_k,v_k)$. The model buyers then make purchase decisions and complete the transaction if meeting budget restriction.
\end{itemize}

\begin{discussion}
       The data owners' contributing willingness is based on risk level if their data is used for training model $\mathcal{M}^m$, while the model buyers' purchasing willingness is based on the usefulness of the model (i.e., model utility). Thus, the two ends have a different standard for the same model tier, which requests the broker to make optimal market decisions to bridge the two ends' different requirements.
\end{discussion}

\subsection{Marketplace Decision Making}\label{subsec:marketDecisionMak}

\subsubsection{Model Training under Data Usage Restriction}
In this part, we propose a brand new versioning strategy, which centers on the data owners' perspective, rather than simply lowering the quality of the model which merely considers the model buyers' payment ability like \cite{DBLP:conf/sigmod/ChenK019}.

\partitle{Versioning}
In general, more data owners' contribution makes the ``manufacturing cost'' of the corresponding model cheaper, because the broker has many alternative choices and tends to dominate the compensation bargain. On the contrary, for the models with higher risk, much fewer data owners are willing to participate, which makes the broker have to pay higher compensation to intrigue more data contribution. With increased manufacturing cost, the product (i.e., the model for sale) price should increase. Thus, we propose to set a version based on the participating risk, which is out of the data owners' perspective. However, since the model buyers, in general, do not care about the risk factor but the model utility, the broker should provide conversion from the risk-tiering standard to the model utility standard for each model tier. To bridge both ends, the broker needs to make market decisions (e.g., set model pricing) constrained by all participating entities' requirements as constraints. As a result, it assures our claim that designing the data marketplace with model-based pricing from end to end is of necessity and importance.

Recall that each model $\mathcal{M}^m$ is associated with model risk $\mathcal{MR}^m$. The versioning is a process where the broker trains a series of models with different model risks: $\mathcal{MR}^1,...,\mathcal{MR}^M$ under the restrictions of all $n$ data owners $\mathcal{DR}_i^m$ and extra compensation requirements $\bm{ec}_i^m$. On the model buyer's end, the broker in our marketplace will provide a utility function $\mathcal{UF}(\mathcal{MR}^m)$ for each risk-based model tier, since the utility is the fundamental model property interested to the model buyers.

In comparison, the versioning strategy of the existing model market \cite{DBLP:conf/sigmod/ChenK019} produces different versions of models by controlling model utility through directly adding noise to the model parameters, in order to suit different model buyers coming with various payment ability. Obviously, their simplified versioning strategy fails to reflect the true ``model manufacturing cost''.

\partitle{Model Utility Maximization with Manufacturing Budget}
Given $\mathcal{O}_i = (\bm{z}_i,\mathcal{DR}_i^m)$ or $\mathcal{O}_i = (\bm{z}_i,\mathcal{DR}_i^m(\bm{ec}_i^m))$, for $i = 1,2,...,n$, the broker trains each model $\mathcal{M}^m$ under data usage restrictions and tries its best to train the best model for each model tier to remain competitive. Let data owner $\mathcal{O}_i$'s preferred model risk be $\bm{risk}_i$, which indicates the highest risk she wants to take (without extra compensation). In this paper, we instantiate the data restriction $\mathcal{DR}_i^m$ function as follows: 1) for the hard restriction case, $\mathcal{DR}_i^m = \mathbb{I}(\mathcal{MR}^m\leq \bm{risk}_i)$, i.e., the data can be used for $\mathcal{M}^m$ only when the model risk $\mathcal{MR}^m$ is lower than the data owner's preferred risk $\bm{risk}_i$; 2) for the negotiable case, $\mathcal{DR}_i^m(\bm{ec}_i^m) = \mathbb{I}(\mathcal{MR}^m\leq \bm{risk}_i) \wedge \mathbb{I}(\bm{ec}_i^m)$, i.e., the data can be used either the model risk is lower than the preferred risk, or the extra compensation is made.

For the simpler hard data usage restriction, the broker trains model $\mathcal{M}^m$ with data \textbf{S}ubset: $\{\bm{S}^m: i\in 1,2,...,n,\ s.t.\ \mathcal{DR}_i^m=1\}$. For the negotiable data usage restriction, under limited manufacturing budget $\mathcal{MB}$, the broker needs to decide whose data worth the extra compensation, so that the utility valuation of the trained model will be maximized for the broker to be competitive in the market. Denote the utility valuation (to be detailed in the next subsection) of $\bm{z}_i$ to model $\mathcal{M}^m$ by $\mathcal{UV}_i^m$. We formalize the subset selection of $\bm{S}^m$ as a training budget constrained utility valuation maximization problem as follows.
\begin{gather}
       \arg\max_{\bm{S}^m \subseteq \{\bm{z}_1,...,\bm{z}_n\}} \sum _{i \in \bm{S}^m} \mathcal{UV}_i^m,\\
       \label{eq.gendealer.constraint}
       s.t.\ (\sum _{i\in \bm{S}^m}(\bm{bc}_i^m+\bm{ec}_i^m(\max\{0,\mathcal{MR}^m-\bm{risk}_i\})) ) \leq \mathcal{MB}^m.
\end{gather}
In the above, the utility value $\mathcal{UV}_i^m$ and the base compensation $\bm{bc}_i^m$ should satisfy certain market design principals, which will be discussed in the following.

\subsubsection{Compensation Allocation}
In this part, we elaborate how Gen-\emph{Dealer} allocates base compensation $\bm{bc}_i^m$ and utility valuation $\mathcal{UV}_i^m$ strategy. Recall that for model $\mathcal{M}^m$, its \textbf{b}ase \textbf{c}ompensation $\bm{bc}_i^m$ and \textbf{e}xtra \textbf{c}ompensation $\bm{ec}_i^m$. The extra compensation is a function of the data owner preferred risk $\bm{risk}_i$ and the model risk $\mathcal{MR}^m$: if $\mathcal{MR}^m\leq\bm{risk}_i$, data owner $\mathcal{O}^i$ will participate the training of $\mathcal{M}^m$ with only base compensation; else if $\mathcal{MR}^m>\bm{risk}_i$, $\bm{ec}_i^m$ is charged with respect to $\mathcal{MR}^m-\bm{risk}_i$, i.e., the broker needs to pay for the extra risk the data owner suffers. Together, $\bm{ec}_i^m$ is a function of $\max\{0,\mathcal{MR}^m-\bm{risk}_i\}$ as shown in Equation (\ref{eq.gendealer.constraint}).

For the base compensation, Gen-\emph{Dealer} allocates it based on the $z_i$'s \textbf{U}tility \textbf{V}alue $\mathcal{UV}_i^m$, where $\mathcal{UV}_i^m$ is based on the (approximate) Shapely value and divides $\bm{bc}_i^m$ according to the relative Shapely value. This way, the true contribution of data owner $\mathcal{O}^i$ to model $\mathcal{M}^m$ can be evaluated and the base compensation is consistent with market design principals. To be practical, efficient approximation algorithms will be utilized.

To summarize, Gen-\emph{Dealer} will allocate the compensation to data owner $\mathcal{O}^i$ for participating model $\mathcal{M}^m$ as $\bm{bc}_i^m+\bm{ec}_i^m(\max\{0,\mathcal{MR}^m-\bm{risk}_i\})) $. The total compensation allocated to data owner $\mathcal{O}^i$ is:
\begin{equation}
    \sum _{m=1}^M \mathbb{I}(i\in\bm{S}^m)\cdot\Big[\bm{bc}_i^m+\bm{ec}_i^m(\max\{0,\mathcal{MR}^m-\bm{risk}_i\}))\Big],
\end{equation}
where $\mathbb{I}(i\in\bm{S}^m)$ is an indicator function for indication whether $z_i$ is in  data subset $\bm{S}^m$ for training $\mathcal{M}^m$.

\subsubsection{Market Survey and Revenue Maximization}\label{subsubsec:marketSurvey}
In this part, we show how to construct the market survey between the broker and the potential model buyers, and how the broker maximizes the revenue based on the market survey.

\partitle{Market Survey}
Prior to release models and prices for sale, the broker will estimate the price for each model through a market survey, which can be done by the broker himself or by third-party companies like consultation service providers. Let the survey size by $K'$, i.e., $K'$ potential model buyers are recruited to provide their purchasing willingness. The survey result will contain $K'$ tuples, one from each survey participant. For the $k^{th}$ survey participant, it provides $(\bm{tm}_k,\bm{v}_k)$, where $\bm{tm}_k\in\{1,...,M\}$ is the target model and $\bm{v}_k$ is the acceptable price of model $\bm{tm}_k$ she is willing to purchase. We note that the survey participants may have an incentive to report lower valuations in order to decrease the price, which can be alleviated by the digital goods auction \cite{DBLP:journals/teco/AlaeiMS14} based on two approaches: random-sampling mechanisms and consensus estimates.

\partitle{Revenue Maximization (Model Pricing)}

With the surveyed purchasing willingness, the broker will price each model in the aim of maximizing revenue and at the same time following the market design principal of arbitrage-free. To do so, the revenue maximization problem is formulated as follows.
\begin{gather}
    \arg\max_{\langle p(\mathcal{MR} ^1),...,p(\mathcal{MR}^M)\rangle} \sum_{m=1}^M \sum_{k=1}^{K'} p(\mathcal{MR}^m) \cdot \mathbb{I}(\bm{tm}_k==m)\cdot\mathbb{I}(p(\mathcal{MR}^m) \leq \bm{v}_k),\\
    s.t.\ p(\mathcal{MR} ^m) + p(\mathcal{MR} ^{m'}) \geq p(\mathcal{MR} ^m + \mathcal{MR} ^{m'}),\ \mathcal{MR} ^{m},\mathcal{MR} ^{m'} \geq 0,\\
    p(\mathcal{MR} ^m) \geq p(\mathcal{MR} ^{m'}),\ \mathcal{MR} ^{m} \geq \mathcal{MR} ^{m'} \geq 0,\\
    p(\mathcal{MR} ^m) \geq 0,\ \mathcal{MR}^m\geq 0,
\end{gather}
where $\mathcal{MR}^m$ is the model risk defined in Section \ref{sub:marketPartForm}, $p(\mathcal{MR}^m)$ is the model price for model $\mathcal{M}^m$ whose model risk is $\mathcal{MR}^m$. In the next section, we will see that this problem is co-NP hard and we will provide an efficient approximation with accuracy bound algorithm there for our DP-\emph{Dealer} instance.

\section{A Differentially Private Data Marketplace Instance and Efficient Approximate Optimization}\label{sec:FrameWorkInstance}
In this section, we propose a concrete realization of Gen-\emph{Dealer}, a differentially private data marketplace with model-based pricing framework DP-\emph{Dealer}. We illustrate the functionalities and restrictions of the data owners and the model buyers in DP-\emph{Dealer} in Sections \ref{sub:DPdataOwner} and \ref{sub:DPmodelBuyer}, respectively. Furthermore, in Section \ref{sub:DPbroker}, we present the broker's functioning in DP-\emph{Dealer} by providing concrete solutions, which are efficient approximate optimization algorithm to make the market practical. Finally, we summarize the complete DP-\emph{Dealer} dynamics in Section \ref{sub:DPcomplete}.

\subsection{Data Owner}\label{sub:DPdataOwner}
We provide a data owner instance by instantiating its risk factor, data restrictions, and extra compensation functions. For the risk factor, we focus on the privacy-preserving issue, which is arguably one of the major concerns limiting individual users from contributing their data. The data owners wish to contribute data for model training with a certain level of privacy in exchange of a fair share of compensation. We follow the differential privacy notion and instantiate both the hard and negotiable data usage restriction cases. Let the model risk factor $\mathcal{MR}^m$ of model $\mathcal{M}^m$ be described by $\epsilon^m$ which corresponds to $\epsilon^m$-differential privacy of the model. Traditional DP system and algorithm designs mostly consider the differential privacy strictness out of the broker' and the model buyers' perspective, which sets it as a tradeoff factor as long as it affects the model utility within a certain level. This overlooks the true privacy demand of the data owners. For those who do consider personalized DP budget, they seldom consider what value the privacy parameter actually means to the data owner. Under such a lack of reward scenario, the data owner still has difficulty evaluating their own privacy demands. Our design allows the data owners to choose their own privacy preference and receives rewards for providing more useful personal information. We believe it is a good starting point for a practical data marketplace with model-based pricing that respects the data owners' privacy demand and incentivizes the data owners for their personal data contribution.

Under the differential privacy risk factor, the three functionalities of data owner $\mathcal{O}_i$ are shown as follows.
\begin{enumerate}[leftmargin=*]
    \item Contributing Data: data owner $\mathcal{O}_i$ contributes her data $z_i=(\bm{x}_i,y_i)$ to the broker;

    \item Setting Usage Restriction: let the personal risk preference $\bm{risk}_i$ be $(\epsilon_{i},\delta)$-DP\footnote{In this paper, we assume the $\delta$ is sufficiently small so that we do not consider its value and composition for the remaining of the paper for convenience.}.

\partitle{Hard DP requirement}
In the first simplified case, each data owner $O_i$ chooses whether her data is allowed for training model $\mathcal{M}^m$ with certain level of privacy restriction by DP parameter $\epsilon_{i}^{prefer}$. That is, data owner $\mathcal{O}^i$ only allows her data to be used for models with DP restrictions stricter than $\epsilon_{i}$, i.e., $\epsilon^m \leq \epsilon_{i}$. Then, the data restriction $\mathcal{DR}_i^m = \mathbb{I}(\epsilon^m \leq \epsilon_{i})$.

\partitle{Negotiable DP requirement}
We further consider a more complicated data owner strategy, where the data owners have more options for making their own trade-offs between compensation and privacy risk. For data owner $\mathcal{O}_i$, in addition to DP requirement $\epsilon_i$, she is also willing to trade some of the privacy for more compensation. To do so, we introduce an extra compensation function $\bm{ec}_i^m(\epsilon_{i},\epsilon^m)$, which pays an extra fraction of the base compensation (allocated based on Shapley value) to compensate for the higher privacy risk. In this case, the data usage restriction function is also a function of the extra compensation: $\mathbb{I}(\bm{ec}_i^m(\epsilon_{i},\epsilon^m))$ which indicates whether the extra compensation has been allocated. Thus, $\mathcal{DR}_i^m(\bm{ec}_i^m) = [\mathbb{I}(\epsilon^m \leq \epsilon_{i}) \wedge \mathbb{I}(\bm{ec}_i^m(\epsilon_{i},\epsilon^m))]$.

\item Receiving Compensation: For both cases, we let the base compensation $\bm{bc}_i^m$ to be proportional to the relative approximated Shapley value (see the broker instance in the following). In addition, for the negotiable case, we introduce the extra compensation function if $\epsilon^m > \epsilon_{i}$ but the broker is willing to use $z_i$ for training $\mathcal{M}^m$ to maximize the model value (subject to the constraint of the manufacturing budget).

\partitle{Extra Compensation Function}
In particular, we present three types of the extra compensation function $\bm{ec}_i^m(\epsilon_{i},\epsilon^m)$: concave, linear, and convex, to model three user inclinations of their personal privacy risks: reserved, balanced, and casual, correspondingly.
\begin{itemize}[leftmargin=*]
    \item linear: $\bm{ec}_i^m(\epsilon_{i},\epsilon^m) = \rho_i^m \bm{bc}_i^m \max\{0,\epsilon^m-\epsilon_{i}\}$;

    \item convex: $\bm{ec}_i^m(\epsilon_{i},\epsilon^m) = \rho_i^m \bm{bc}_i^m (\max\{0,\epsilon^m-\epsilon_{i}\})^2$;

    \item concave: $\bm{ec}_i^m(\epsilon_{i},\epsilon^m) = \rho_i^m \bm{bc}_i^m (\max\{0,\epsilon^m-\epsilon_{i}\})^{\frac{1}{2}}$;
\end{itemize}
\end{enumerate}
For the ease of presentation, we use $\bm{ec}_i^m$ to replace $\bm{ec}_i^m(\epsilon_{i},\epsilon^m)$ to express the extra compensation of data owner $\mathcal{O}_i$ on model $\mathcal{M}^m$ in the following.

\begin{discussion}
       For all cases, each data owner $\mathcal{O}_i$ has the maximum total potential privacy leakage $\sum_{m=1}^{M} \epsilon_i^m$. In this work, we assume the data owners do not have too much information about the data marketplace except the information given by the broker. Thus, they invariably invest their total privacy budget to each of the $M$ models, which can be suboptimal for certain owners. In the future, we will consider more informed data owners, who have not only more knowledge about the data marketplace, e.g., the demand for each type of model, but also the quality and privacy restrictions of other data owners. With the additional market information, data owners can allocate their total privacy leakage more intelligently by investing their privacy budget towards models returning them more compensation.
\end{discussion}

\begin{assumption}
The goal of adding DP noise in models is to limit what can be inferred from the models about individual training data tuples. Therefore, it is better for the broker to support the relationship between DP parameter $\epsilon$ and what can be inferred from the models to the data owners. We note that such a relationship can be implemented by \cite{DBLP:conf/uss/Jayaraman019}.
\end{assumption}

\subsection{Model Buyer}\label{sub:DPmodelBuyer}
The model buyers in DP-\emph{Dealer} have the same functionalities with the model buyers in Gen-\emph{Dealer} shown in Section \ref{sub:marketPartForm} when considering the differential privacy instantiation.

\begin{assumption}
The arbitrage-free property of the models is established in terms of the differential privacy budget.
\end{assumption}

\begin{remark}
In practice, the model buyers can buy a couple of weak models and convert those weak models to strong ones by employing some machine learning techniques such as ensemble learning, bagging and boosting. However, it may be infeasible to formally characterize how model combinations behave in terms of the model utility. Therefore, instead of ensuring the models to satisfy arbitrage-free in terms of model utility, we ensure that the models satisfy arbitrage-free in terms of DP parameter. 
\end{remark}

\subsection{Broker}\label{sub:DPbroker}
In this part, we present the broker's functioning in the differentially private data marketplace by providing concrete solutions for selecting optimal training subsets with budget constraints, the market survey to potential model buyers, and pricing models for revenue maximization with an arbitrage-free guarantee.

\subsubsection{Selecting Optimal Training Subsets with Budget Constraint}
Given the training data along with the data owners' privacy and extra compensation functions, the broker aims to train the highest valued model for each price tier with the constraint on the privacy and manufacturing/compensation budget. According to different data owners' requirements, the broker has two types of workflows.

\partitle{Processing the Hard DP Restriction}
In this case, for model tier $\epsilon^{m}$, the broker is allowed to release model $\mathcal{M}^m$ trained strictly with a subset $\bm{S}^m$ within the data owners $\mathcal{O}_i$, where $\epsilon_{i} \leq \epsilon^m$. We formalize the optimization problem as follows.
\begin{equation}\label{equ:WorHardDPrestr}
    \arg\max_{\bm{S}^m} \sum_{i\in\bm{S}^m}\mathcal{SV}_i^m,\ s.t.\ \sum _{i\in\bm{S}^m} \bm{bc}_i^m \leq \mathcal{MB}^m
\end{equation}
We omit the solutions for Equation (\ref{equ:WorHardDPrestr}) because it is a special case of the following Equation (\ref{equ:WorNegotiableDPrestr}).

\partitle{Processing the Negotiable DP Restriction}
A more practical case is the negotiable DP restriction, where the broker has the option to decide whether to intrigue high quality data owners to lower their privacy restriction with extra compensation. We formalize it as the following Budget Constrained Maximum Value Problem (BCMVP) on model $\mathcal{M}^m$:
\begin{equation}\label{equ:WorNegotiableDPrestr}
    \arg\max_{\bm{S}^m} \sum_{i\in\bm{S}^m}\mathcal{SV}_i^m,\ s.t.\ \sum _{i\in\bm{S}^m} (\bm{bc}_i^m+\bm{ec}_i^m) \leq \mathcal{MB}^m
\end{equation}
where $\mathcal{MB}^m$ is the manufacturing budget of model $\mathcal{M}^m$. In essence, we reallocate the payment of lower valued data owners to become the extra compensation of higher valued data owners.


The above problem is difficult to be exactly solved. In fact, we prove the problem is NP-hard. Given this NP-hard complexity, we then present three approximation algorithms. First, we present a pseudo-polynomial time algorithm using dynamic programming technique. Then, we present a fully polynomial-time approximation scheme with the worst case bound if each data owner's compensation is not too large. Finally, we propose an enumeration guess based polynomial time approximation algorithm with the worst case bound by relaxing the compensation constraint, which uses the pseudo-polynomial time algorithm as a subroutine.

\partitle{NP-hardness proof}
We prove that BCMVP is NP-hard by showing that the well-known partition problem is polynomial time reducible to BCMVP.

\begin{definition}(\textbf{Decision Version of BCMVP})
Given a set $\bm{S}$ of $n$ data owners with their corresponding privacy compensation $\bm{bc}_1^m+\bm{ec}_1^m, \bm{bc}_2^m+\bm{ec}_2^m, ..., \bm{bc}_n^m+\bm{ec}_n^m$ and Shapley value $\mathcal{SV}_1^m, \mathcal{SV}_2^m, ..., \mathcal{SV}_n^m$, the decision version of BCMVP has the task of deciding whether there is a subset $\bm{S}_1 \subseteq \bm{S}$ such that $\sum_{i\in \bm{S}_1}\bm{bc}_i^m+\bm{ec}_i^m\leq B$ and $\sum_{i\in \bm{S}_1} \mathcal{SV}_i^m \geq V$.
\end{definition}

\begin{definition}(\textbf{Decision Version of Partition Problem})
Given a set $\bm{S}$ of $n$ positive integer values $v_1, v_2, ...,v_n$, the decision version of partition problem has the task of deciding whether the given set S can be partitioned into two subsets $\bm{S}_1$ and $\bm{S}_2$ such that the sum of the integers in $\bm{S}_1$ equals the sum of the integers in $\bm{S}_2$.
\end{definition}

\begin{theorem}
The decision version of BCMVP is an NP-hard problem.
\end{theorem}

\begin{proof}
We show that there exists a polynomial reduction by proving that there exists a subset $\bm{S}_1 \subseteq \bm{S}$ such that $\sum_{i\in \bm{S}_1}\bm{bc}_i^m+\bm{ec}_i^m\leq B$ and $\sum_{i\in \bm{S}_1} \mathcal{SV}_i^m \geq V$ if and only if there is a partition $\bm{S}_1$ and $\bm{S}_2$ such that the sum of the integer values in $\bm{S}_1$ equals the sum of the integer values in $\bm{S}_2$. We construct the polynomial reduction as follows. Consider the following instance of BCMVP: $\bm{bc}_i^m+\bm{ec}_i^m=v_i$ and $\mathcal{SV}_i^m=v_i$ for $i=1,2,...,n$, and $B=V=\frac{1}{2}\sum_{i=1}^nv_i$.

We show the reduction as follows.

(1) If there exists a partition $\bm{S}_1$ and $\bm{S}_2$ such that the sum of the integer values in $\bm{S}_1$ equals to the sum of the integers in $\bm{S}_2$, there exists $\bm{S}_1$ and $\bm{S}_2$ such that $\sum_{i\in \bm{S}_1}v_i=\sum_{i\in \bm{S}_2}v_i=\frac{1}{2}\sum_{i=1}^nv_i$. We choose the set of data owners $\bm{S}_1$ in BCMVP and we have $\sum_{i\in \bm{S}_1}\bm{bc}_i^m+\bm{ec}_i^m=\sum_{i\in \bm{S}_1}v_i=\frac{1}{2}\sum_{i=1}^nv_i=B$ and $\sum_{i\in \bm{S}_1} \mathcal{SV}_i^m=\sum_{i\in \bm{S}_1}v_i=\frac{1}{2}\sum_{i=1}^nv_i = V$. Therefore, we know that there exists a subset $\bm{S}_1 \subseteq \bm{S}$ such that $\sum_{i\in \bm{S}_1}\bm{bc}_i^m+\bm{ec}_i^m\leq B$ and $\sum_{i\in \bm{S}_1} \mathcal{SV}_i^m \geq V$.

(2) If there exists a subset $\bm{S}_1 \subseteq \bm{S}$ such that $\sum_{i\in \bm{S}_1}\bm{bc}_i^m+\bm{ec}_i^m\leq B$ and $\sum_{i\in \bm{S}_1} \mathcal{SV}_i^m \geq V$, we partition the set $\bm{S}$ into $\bm{S}_1$ and $\bm{S}_2=\bm{S}-\bm{S}_1$. We have $\sum_{i\in \bm{S}_1}\bm{bc}_i^m+\bm{ec}_i^m=\sum_{i\in \bm{S}_1}v_i\leq B=\frac{1}{2}\sum_{i=1}^nv_i$ and $\sum_{i\in \bm{S}_1} \mathcal{SV}_i^m=\sum_{i\in \bm{S}_1}v_i \geq V=\frac{1}{2}\sum_{i=1}^nv_i$. This implies that $\sum_{i\in \bm{S}_1}v_i =\frac{1}{2}\sum_{i=1}^nv_i$. We also have $\sum_{i\in \bm{S}_2}v_i =\sum_{i=1}^nv_i-\frac{1}{2}\sum_{i=1}^nv_i=\frac{1}{2}\sum_{i=1}^nv_i$. Therefore, there exists a partition $\bm{S}_1$ and $\bm{S}_2$ such that $\sum_{i\in \bm{S}_1}v_i=\sum_{i\in \bm{S}_2}v_i=\frac{1}{2}\sum_{i=1}^nv_i$.
\end{proof}

\partitle{Pseudo-polynomial time algorithm}
We present a pseudo-polynomial time algorithm for BCMVP. Pseudo-polynomial means that our algorithm has the polynomial time complexity in terms of $\mathcal{MB}^m$ rather than the number of data owners $n$. We divide $\mathcal{MB}^m$ into $\lceil \frac{\mathcal{MB}^m}{a}\rceil$ parts, where $a$ is the greatest common divisor in $\bm{bc}_i^m+\bm{ec}_i^m$ for all $i=1,2,...,n$. We define $\mathcal{SV}[i,j]$ as the maximum BCMVP that can be attained with compensation budget $\leq j\times a$ by only using the first $i$ data owners. The detailed algorithm is shown in Algorithm \ref{Alg:DynamicProg}. In Line 5, if the compensation budget is not enough, we do not need to consider the $i^{th}$ data owner. Otherwise, we can take $\mathcal{O}_i$ if we can get more value by replacing some data owners from $\mathcal{O}_1,...,\mathcal{O}_{i-1}$ in Line 8.

\begin{algorithm}[thb] \caption{Pseudo-polynomial time algorithm for BCMVP.}\label{Alg:DynamicProg}
\SetKwInOut{Input}{input}\SetKwInOut{Output}{output}

\Input{$\bm{bc}_i^m+\bm{ec}_i^m$, $\mathcal{MB}^m$, and $\mathcal{SV}_i^m$ for $i=1,2,...,n$.}
\Output{$\bm{S}^m$.}

\For{j=0:a:$\mathcal{MB}^m$}{
$\mathcal{SV}[0,j]=0$\;
}

\For{i =1 to n}{
    \For{j=0:a:$\mathcal{MB}^m$}{
    \If{$\bm{bc}_i^m+\bm{ec}_i^m>j\times a$}{
    $\mathcal{SV}[i,j]=\mathcal{SV}[i-1,j]$\;}
    \Else{
    $\mathcal{SV}[i,j]=max\{\mathcal{SV}[i-1,j],\mathcal{SV}[i-1,j\times a-\bm{bc}_i^m+\bm{ec}_i^m]+\mathcal{SV}_i^m\}$\;}
    }
}
backtrack from $\mathcal{SV}[n,\lceil \frac{\mathcal{MB}^m}{a}\rceil]$ to $\mathcal{SV}[1,0]$ to find the selected $\mathcal{O}_i$\;
\end{algorithm}

\partitle{Polynomial-time approximation algorithm}
The time cost of the proposed pseudo-polynomial time algorithm in Algorithm \ref{Alg:DynamicProg} is extremely dominated by the compensation budget. We propose a simple yet efficient polynomial-time approximation algorithm in Algorithm \ref{Alg:greedy}, which is not sensitive to the compensation budget. We sort the data owners in decreasing order of Shapley value per compensation budget $\frac{\mathcal{SV}_i^m}{\bm{bc}_i^m+\bm{ec}_i^m}$ in Line 3. In Lines 6-8, we proceed to take the data owners, starting with as high as possible of $\frac{\mathcal{SV}_i^m}{\bm{bc}_i^m+\bm{ec}_i^m}$ until there is no budget. We also present a lower bound for Algorithm \ref{Alg:greedy} in Theorem \ref{the:polynomialTimeGreedy}, where $MAX$ is the maximum value that we can obtain in function (\ref{equ:WorNegotiableDPrestr}).

\begin{algorithm}[thb] \caption{Polynomial-time approximation algorithm for BCMVP.}\label{Alg:greedy}
\SetKwInOut{Input}{input}\SetKwInOut{Output}{output}

\Input{$\bm{bc}_i^m+\bm{ec}_i^m$, $\mathcal{MB}^m$, and $\mathcal{SV}_i^m$ for $i=1,2,...,n$.}
\Output{$\bm{S}^m$.}

\For{i=1 to n}{
compute $\frac{\mathcal{SV}_i^m}{\bm{bc}_i^m+\bm{ec}_i^m}$\;
}
sort $\frac{\mathcal{SV}_i^m}{\bm{bc}_i^m+\bm{ec}_i^m}$ for $i=1,2,...,n$ in decreasing order and denote as $\frac{\mathcal{SV}_1^m}{\bm{bc}_1^m+\bm{ec}_1^m}\geq \frac{\mathcal{SV}_2^m}{\bm{bc}_2^m+\bm{ec}_2^m}\geq ... \geq \frac{\mathcal{SV}_n^m}{\bm{bc}_n^m+\bm{ec}_n^m}$\;
B=0\;
i=1\;
\While{$B\leq \mathcal{MB}^m$}
{
add $\bm{bc}_i^m+\bm{ec}_i^m$ to B\;
i=i+1\;}
\Return the corresponding $\mathcal{O}_i$ of those $\bm{bc}_i^m+\bm{ec}_i^m$ in B\;
\end{algorithm}

\begin{theorem}\label{the:polynomialTimeGreedy}
If for all $i$, $\bm{bc}_i^m+\bm{ec}_i^m\leq \zeta \mathcal{MB}^m$, Algorithm \ref{Alg:greedy} has a lower bound guarantee $(1-\zeta)MAX$.
\end{theorem}
\begin{proof}
We set $\bm{bc}_k^m+\bm{ec}_k^m$ as the first data that is not accepted in Algorithm \ref{Alg:greedy}, i.e., we choose the corresponding data owners of $\bm{bc}_1^m+\bm{ec}_1^m, \bm{bc}_2^m+\bm{ec}_2^m, ..., \bm{bc}_{k-1}^m+\bm{ec}_{k-1}^m$. For $1\leq i\leq k$, we have $\frac{\mathcal{SV}_i^m}{\bm{bc}_i^m+\bm{ec}_i^m}\geq \frac{\mathcal{SV}_k^m}{\bm{bc}_k^m+\bm{ec}_k^m}$.

$\Rightarrow~ \mathcal{SV}_i^m\geq (\bm{bc}_i^m+\bm{ec}_i^m)\frac{\mathcal{SV}_k^m}{\bm{bc}_k^m+\bm{ec}_k^m}$

$\Rightarrow~ \mathcal{SV}_1^m+\mathcal{SV}_2^m+...+\mathcal{SV}_k^m\geq (\bm{bc}_1^m+\bm{ec}_1^m+\bm{bc}_2^m+\bm{ec}_2^m+...+\bm{bc}_k^m+\bm{ec}_k^m)\frac{\mathcal{SV}_k^m}{\bm{bc}_k^m+\bm{ec}_k^m}$

Because we set $\bm{bc}_k^m+\bm{ec}_k^m$ as the first data that is not accepted, i.e., $\bm{bc}_1^m+\bm{ec}_1^m+\bm{bc}_2^m+\bm{ec}_2^m+...+\bm{bc}_k^m+\bm{ec}_k^m> \sum_{i=1}^n\bm{bc}_i^m$, we have

$\Rightarrow~ \mathcal{SV}_k^m\leq (\mathcal{SV}_1^m+\mathcal{SV}_2^m+...+\mathcal{SV}_k^m)\frac{\bm{bc}_k^m+\bm{ec}_k^m}{\sum_{i=1}^n\bm{bc}_i^m}$

$\Rightarrow~ \mathcal{SV}_k^m\leq \zeta(\mathcal{SV}_1^m+\mathcal{SV}_2^m+...+\mathcal{SV}_k^m)$

$\Rightarrow~ \mathcal{SV}_k^m\leq \frac{\zeta(\mathcal{SV}_1^m+\mathcal{SV}_2^m+...+\mathcal{SV}_{k-1}^m)}{1-\zeta}$

Because $\mathcal{SV}_1^m+\mathcal{SV}_2^m+...+\mathcal{SV}_k^m\geq MAX$, we have $\mathcal{SV}_1^m+\mathcal{SV}_2^m+...+\mathcal{SV}_{k-1}^m\geq (1-\zeta)MAX$. Therefore, Algorithm \ref{Alg:greedy} has a lower bound guarantee $(1-\zeta)MAX$.
\end{proof}

\begin{lemma}\label{lemma:opt}
There are at most $\lceil \frac{1}{\alpha}\rceil$ data owners having compensation $\bm{bc}_i^m+\bm{ec}_i^m$ such that their corresponding Shapley value $\mathcal{SV}_i^m$ is at least $\alpha MAX$ in any optimal solution.
\end{lemma}
Lemma \ref{lemma:opt} is easy to see, otherwise, the optimal solution value is larger than $MAX$, which is a contradiction.

\partitle{Enumeration guess based polynomial time approximation algorithm}
Although Algorithm \ref{Alg:greedy} can achieve ($1-\zeta)MAX$, the requirement of  $\bm{bc}_i^m+\bm{ec}_i^m\leq \zeta \mathcal{MB}^m$ is too strict. We present another algorithm with the same worst case bound but without the above requirement. Let $\alpha\in (0,1)$ be a fixed constant and $h=\lceil \frac{1}{\alpha}\rceil$. We will try to guess the $h$ most profitable data owners in an optimal solution and compute the rest greedily as in Algorithm \ref{Alg:greedy}. The detailed algorithm is shown in Algorithm \ref{Alg:GuessGreedy} as follows. We first enumerate all the subsets with data owner size $\leq h$ in Lines 1-3. We delete those subsets with higher compensation budget than $\mathcal{MB}^m$ in Lines 4-6. In Lines 7-10, for each remaining subset, we call Algorithm \ref{Alg:greedy} to maximize the value with the remaining budget after taking the $\leq h$ data owners.

\begin{algorithm}[thb] \caption{Guess and Polynomial-time approximation algorithm for BCMVP.}\label{Alg:GuessGreedy}
\SetKwInOut{Input}{input}\SetKwInOut{Output}{output}

\Input{$\bm{bc}_i^m+\bm{ec}_i^m$, $\mathcal{MB}^m$, and $\mathcal{SV}_i^m$ for $i=1,2,...,n$.}
\Output{$\bm{S}^m$.}

\For{i=1 to h}{
choose $i$ data owner(s) to compose a subset $\bm{S}'$\;
}
we have $\sum_{i=1}^h {n\choose i}$ such subsets\;
\For{j=1 to $\sum_{i=1}^h {n\choose i}$}{
compute the compensation budget of the data owners in $\bm{S}'$\;
delete those $\bm{S}'$ if their compensation budget is larger than $\mathcal{MB}^m$\;
}

we have $r$ remaining subsets $\bm{S}'_1,\bm{S}'_2,...,\bm{S}'_r$\;
\For{each subset $\bm{S}'_j$, $j=1,2,...,r$}{
let $\mathcal{O}_a$ be the data owner with the least Shapley value in $\bm{S}'_j$, remove all data owners in $\bm{S}_j-\bm{S}'_j$ if their Shapley value is larger than $\mathcal{SV}_a^m$ and get a new subset $\bm{S}''_j$\;
run Algorithm \ref{Alg:greedy} in $\bm{S}''_j$ with remaining compensation budget $\mathcal{MB}^m-\sum_{i=1}^{|\bm{S}'_j|} (\bm{bc}_i^m+\bm{ec}_i^m)$\;}
\Return the data owners in $\bm{S}'_j$ and $\bm{S}''_j$, where $\bm{S}'_j$ and $\bm{S}''_j$ have the highest Shapley value among $j=1,2,...,r$\;
\end{algorithm}

\begin{theorem}
  Algorithm \ref{Alg:GuessGreedy} runs in $O(n^{\lceil \frac{1}{\alpha}\rceil})$ time with ($1-\alpha$)MAX worst case bound.
\end{theorem}
\begin{proof}
For the time complexity, we have at most $\sum_{i=1}^h {n\choose i}$ subsets $\bm{S}'$ after deleting those subsets if their compensation budget is larger than $\mathcal{MB}^m$. That is, we have at most $n^h$ different subsets $\bm{S}'$. For each subset $\bm{S}'$, the greedy Algorithm \ref{Alg:greedy} only requires linear time to handle the remaining data owners. Therefore, the total time cost for Algorithm \ref{Alg:GuessGreedy} is $O(n^{\lceil \frac{1}{\alpha}\rceil+1})$.

For the worst case approximation bound, we assume subset $\bm{S}'$ in the optimal solution has exact $h$ data owners. We note that subset $\bm{S}'$ in the optimal solution may have $\leq h$ data owners, but it is easy to see that this does not affect the complexity analysis. If the number of data owners in the optimal solution is less than $h$, the optimal solution will be included in $\bm{S}'$. In the following, we discuss the case that the number of data owners in the optimal solution is larger than $h$.

We have $h+k$ data owners $\mathcal{O}_1,...,\mathcal{O}_h,\mathcal{O}_{h+1},...,\mathcal{O}_{h+k-1},\mathcal{O}_{h+k}$ that need to be considered, where $\mathcal{O}_1,...,\mathcal{O}_h$ are the data owners in subset $\bm{S}'$, $\mathcal{O}_{h+i}$ is the $i^{th}$ data owner with the highest $\frac{\mathcal{SV}_i^m}{\bm{bc}_i^m+\bm{ec}_i^m}$ in $\bm{S}''$. $\mathcal{O}_{h+k}$ is the data owner with the highest $\frac{\mathcal{SV}_i^m}{\bm{bc}_i^m+\bm{ec}_i^m}$ rejected by the greedy algorithm of Algorithm \ref{Alg:greedy}. Let $MAX'$ be the optimal value for the data owners in $\bm{S}''$. Therefore, we have $\mathcal{SV}(\bm{S}'')+\mathcal{SV}_{h+k}^m\geq MAX'$.

$\Rightarrow~\mathcal{SV}(\bm{S}'')\geq MAX'-\mathcal{SV}_{h+k}^m$

Based on Lemma \ref{lemma:opt}, there are at most $\lceil \frac{1}{\alpha}\rceil$ data owners having compensation $\bm{bc}_i^m+\bm{ec}_i^m$ such that their corresponding Shapley value $\mathcal{SV}_i^m$ is at least $\alpha MAX$ in any optimal solution, and those $\lceil \frac{1}{\alpha}\rceil$ data owners are already pruned in Line 9. Therefore, we have $\mathcal{SV}_{h+k}^m\leq \alpha MAX$.

$\Rightarrow~\mathcal{SV}(\bm{S}'')\geq MAX'-\alpha MAX$

$\Rightarrow~\mathcal{SV}(\bm{S}')+ \mathcal{SV}(\bm{S}'')\geq MAX'+\mathcal{SV}(\bm{S}')-\alpha MAX$

$\Rightarrow~\mathcal{SV}(\bm{S}')+ \mathcal{SV}(\bm{S}'')\geq MAX-\alpha MAX$

That is, Algorithm \ref{Alg:GuessGreedy} has the worst case bound $(1-\alpha)MAX$.
\end{proof}

\subsubsection{Market Survey to Potential Model Buyers}\label{subsubsec:MStoMB}
In the previous subsection, the broker requires the model budget as a constraint to manufacture the models, which is before the models are available to the model buyers. To acquire the budget variable, a common practice is to perform a market survey to collect purchasing willingness from potential model buyers. That is, the broker presents a series of potential models along with their performance estimation to the potential model buyers, who then provide which model they are willing to purchase and at what price. Based on the survey result, the broker can estimate the budget by solving a revenue maximization problem in the next subsection. The market survey stage is sometimes the earliest stage among the overall market dynamics.

In the following part, we propose a survey approach by overcoming two difficulties. First, the broker encounters different standards for categorizing the tier of each model. During the manufacture, each data owner uses a differential privacy budget to differentiate the model tier, the model buyers, however, are unlikely to care about the restriction in the privacy of the model they purchase. On the contrary, it is the model prediction performance that they pay attention to. Thus, to sell models to the model buyers, the broker needs to transit the $\epsilon^m$-DP based model description to the prediction performance based model description. This raises the second difficulty: the Shapley value based utility measure is not available at the survey stage (the data may not even been collected yet). To overcome both difficulties, we utilize a common estimation of utility for the DP ERM models, which converts the DP parameter to a general excess population loss by assuming all data samples are identically independently distributed. It also reveals the relation between the number of training samples and the utility estimation, which provides a guide to the data collection.


For training each model $\mathcal{M}^m$ subject to DP restriction $\epsilon^m$, the broker uses a subset $\bm{S}^m$ out of all data available in the market $\bm{D}$, whose data owners have $ \epsilon \geq \epsilon^m $. Recall that the full dataset $\bm{D}$ is from distribution $\mathcal{D}$, i.e., $\bm{D}\sim \mathcal{D}$. For the model buyers, they care about the performance of the model on their prediction tasks. That is, for $\bm{z}_{predict} = (\bm{x}_{predict},y_{predict})$, where $\bm{z}_{predict}\sim \mathcal{D}$, the broker estimates the value of a particular tier of model by estimating $l(\bm{w}^m,\bm{z}_{predict})$. To formalize it, we utilize the notion of \emph{population loss} as follows,
\begin{definition}(Population Loss)
\begin{equation}
    \mathcal{L}(\bm{w};\mathcal{D}) := \mathbb{E}_{\bm{z}\sim \mathcal{D}}[l(\bm{w},\bm{z})],
\end{equation}
where the expectation is over the distribution of the data.
\end{definition}
Thus, it measures the expected prediction loss of a model $\mathcal{M}^m$ when given the output $\bm{w}^m$. With the population risk, the broker provides the maximum discrepancy between the ideal model with model parameter $\bm{w}^*$ and the DP one for sale $\bm{w}^m_{DP}$. The following \emph{excess population loss} notion formalizes this discrepancy,
\begin{definition} (Excess Population Loss \cite{bassily2019private})\label{def:excessPopuLoss}
\begin{equation}
    \Delta\mathcal{L}(\mathcal{A}^m_{Algorithm \ref{Alg:DPtraining}};\bm{S}^m) := \mathbb{E}[\mathcal{L}(\bm{w}^m_{DP};\mathcal{D}) - \mathcal{L}(\bm{w}^*;\mathcal{D})],
\end{equation}
where $\bm{w}^* = \arg\min \mathcal{L}(\bm{w},\mathcal{D}),\ s.t.\ \bm{w}\in\mathcal{W}$ (model parameter space), $\mathcal{A}^m_{Algorithm\ref{Alg:DPtraining}}$ denotes the DP algorithm in Algorithm \ref{Alg:DPtraining} for training model $\mathcal{M}^m$ with DP restricted dataset $\bm{S}^m$, and the expectation is taken over the randomness of Algorithm \ref{Alg:DPtraining}.
\end{definition}

The specific accuracy measure can vary from application to application and be chosen according to the model buyers. The order of the population loss is more universal and general than a specific choice of accuracy metric. Also, we believe the order makes more sense than a particular number reported on a given testing dataset. Thus, the excess population risk serves as a good estimation of utility at the market survey stage.

The following theorem from \cite{bassily2019private} provides an estimate for the excess population loss for model $\mathcal{M}^m$.

\begin{theorem}\label{the:excessPopuLoss}
   Under certain conditions, the excess population loss for the output $\bm{w}_{DP}^m$ of the objective perturbation based training algorithm $\mathcal{A}^m_{Algorithm \ref{Alg:DPtraining}}$ is
  \begin{equation}
      \Delta\mathcal{L}(\mathcal{A}^m_{Algorithm \ref{Alg:DPtraining}};\epsilon^m,\bm{S}^m) = O(\max\{\frac{1}{\sqrt{|\bm{S}^m|}},\frac{\sqrt{d\log(1/\delta)}}{\epsilon ^m |\bm{S}^m|}\}).
  \end{equation}
\end{theorem}

From the above theorem, we can see that the model price dominated by the excess population loss is an increasing function of $\epsilon$ because of two reasons: 1) from the model buyer perspective, the larger the $\epsilon$ (i.e., the looser the privacy restriction), the higher the model utility, which will better meet the model buyers' usage. It is consistent with the common belief that a better product will cost more. 2) from the data owner perspective, less owners are willing to allow their data to be used for looser privacy protected models, which leads to less contributors without extra compensation and the broker will have increased manufacturing cost in order to recruit more data owners. It is also consistent with the common belief that the product with higher manufacturing cost should be more expensive.

In addition to the privacy budget, the training set size provides an important role. First, in the previous model based pricing paper \cite{DBLP:conf/sigmod/ChenK019}, the authors claim that the versioning of the various quality of models does not result into the extra cost for the broker, is however not true according to the $|\bm{S}^m|$ in the above theorem. In fact, the manufacturing cost actually increases for higher tier models (i.e., the one with a larger $\epsilon ^m$). That is, less data owners are will to contribute their data for models with larger $\epsilon ^m$ without extra compensation. Thus, the broker has to spend more manufacturing cost on recruiting more data owners for model training (i.e., for extra compensation), otherwise the smaller $\bm{S}^m$ will lead to lower model utility despite the increasing $\epsilon$. Second, the $|\bm{S}^m|$ in the above theorem also provides a good guidance to the broker on the data collection phase, i.e. at least how much data owners the broker has to engage with.

\subsubsection{Pricing Models for Revenue Maximization with Arbitrage-free Guarantee}
Before the market survey, the broker provides the excess population risk estimation for each DP-model to the $K'$ survey participants, who are potential model buyers and are interested in the model performance rather than the privacy risk. Each participant $B_k$ is asked to provide which model they want to purchase (target) $\bm{tm}_k$, and at what price $\bm{v}_k$. To make the arbitrage-free pricing $p(\epsilon^m)$ for model $\mathcal{M}^m$ with differential privacy $\epsilon^m$, the objective function for the broker is
\begin{gather}
    \arg\max_{\langle p(\epsilon ^1),...,p(\epsilon^M)\rangle} \sum_{m=1}^M \sum_{k=1}^{K'} p(\epsilon^m) \cdot \mathbb{I}(\bm{tm}_k==m)\cdot\mathbb{I}(p(\epsilon^m) \leq \bm{v}_k),\\ \label{ObFun:RM}
    s.t.\ p(\epsilon ^m) + p(\epsilon ^{m'}) \geq p(\epsilon ^m + \epsilon ^{m'}),\ \epsilon ^{m},\epsilon ^{m'} \geq 0,\\ \label{ObFun:subadd}
    p(\epsilon ^m) \geq p(\epsilon ^{m'}) \geq 0,\ \epsilon ^{m} \geq \epsilon ^{m'} \geq 0
\end{gather}
We refer this problem as Revenue Maximization ($\mathcal{RM}$) problem and the optimal revenue for $\mathcal{RM}$ as $OPT(\mathcal{RM})$. We use $(m,sp^m[j])$ to denote the $j^{th}$ lowest survey price point in the $m^{th}$ model. For example, in Figure \ref{fig:rmalgorithm}, we have six survey participants shown in black disk $(1,sp^1[1]=1)$, $(1,sp^1[2]=4)$, $(2,sp^2[1]=3)$, $(2,sp^2[2]=7)$, $(3,sp^3[1]=5)$, and $(3,sp^3[2]=8)$. In general, the smaller the $\epsilon$, the stricter the privacy restriction, which results to lower model price. The reason for this price trend will be revealed by the next subsection.

Given a set of survey price points $(m,sp^m[1])$ for $m=1,2,...,M$, does there exist a pricing function $p(\epsilon^m)$ such that 1) is positive, monotone, and subadditive; and 2) ensures $p(\epsilon^m)=sp^m[1]$ for all $m=1,2,...,M$, which is a co-NP hard problem \cite{DBLP:conf/sigmod/ChenK019}. It is easy to see that this co-NP hard problem is a special case of our $\mathcal{RM}$ problem, i.e., there is only one survey price point for each model. Therefore, it is suspected that there is no polynomial-time algorithm for the proposed $\mathcal{RM}$ problem.

In order to overcome the hardness of the original optimization $\mathcal{RM}$ problem, we seek to approximately solve the problem by relaxing the subadditivity constraint. We relax the constraint of $p(\epsilon^m)+p(\epsilon^{m'})\geq p(\epsilon^m+\epsilon^{m'})$ in Equation (\ref{ObFun:subadd}) to
\begin{equation}
    p(\epsilon^m)/\epsilon^m\geq p(\epsilon^{m'})/\epsilon^{m'}
\end{equation}
which still satisfies the requirement of arbitrage-free. We refer this relaxed problem as Relaxed Revenue Maximization ($\mathcal{RRM}$) problem. Generally speaking, we want to make sure that the unit price for large purchases is smaller than or equals to the unit price for small purchases, which is practical in the real marketplaces.

In the following, we show the maximum revenue for $\mathcal{RRM}$, $OPT(\mathcal{RRM})$, has a lower bound with respect to the maximum revenue for $\mathcal{RM}$, $OPT(\mathcal{RM})$.
\begin{theorem}
  The maximum revenue for $\mathcal{RM}$ has the following relationship with the maximum revenue for $\mathcal{RRM}$
  \begin{displaymath}
 OPT(\mathcal{RRM})\geq  OPT(\mathcal{RM})/2
  \end{displaymath}
\end{theorem}

\begin{proof}
Given a feasible solution $p_{\mathcal{RM}}$ of the revenue maximization problem, we construct a solution $p_{\mathcal{RRM}}$ such that for all $m>0$, $p_{\mathcal{RRM}}(\epsilon^ m)=\epsilon^m\times min_{0<x\leq m}\{p_{\mathcal{RM}}(\epsilon^x)/\epsilon^x\}$, where $p_{\mathcal{RRM}}(\epsilon^ m)$ is the price in $\mathcal{RRM}$ for model $\mathcal{M}^m$ with privacy parameter $\epsilon^m$. Let $0<m\leq m'$. We show that $p_{\mathcal{RRM}}$ is a feasible solution of the relaxing maximization problem as follows.

We prove that $p_{\mathcal{RRM}}$ satisfies the monotone property as follows. Let $x'_{min}=\arg min_{0<x\leq m'}\{p_{\mathcal{RM}}(\epsilon^x)/\epsilon^x\}$. We have two cases for $x'_{min}$, $0<x'_{min}\leq m$ and $m<x'_{min}\leq m'$. For the first case $0<x'_{min}\leq m$, we have $min_{0<x\leq m'}\{p_{\mathcal{RM}}(\epsilon^x)/\epsilon^x\}=min_{0<x\leq m}\{p_{\mathcal{RM}}(\epsilon^x)/\epsilon^x\}$ because $x'_{min}$ lies in the range of $(0,m]$. And then we have $p_{\mathcal{RRM}}(\epsilon^{m'})=\epsilon^{m'}\times min_{0<x\leq m'}\{p_{\mathcal{RM}}(\epsilon^x)/\epsilon^x\}\geq \epsilon^m\times min_{0<x\leq m}\{p_{\mathcal{RM}}(\epsilon^x)/\epsilon^x\}=p_{\mathcal{RRM}}(\epsilon^m)$. For the second case $m<x'_{min}\leq m'$, we have $p_{\mathcal{RRM}}(\epsilon^m)=\epsilon^m\times min_{0<x\leq m}\{p_{\mathcal{RM}}(\epsilon^x)/\epsilon^x\}\leq \epsilon^m\times \{p_{\mathcal{RM}}(\epsilon^x)/\epsilon^x\}_{x=m}=\epsilon^m\times \{p_{\mathcal{RM}}(\epsilon^m)/\epsilon^m\} = p_{\mathcal{RM}}(\epsilon^m)<p_{\mathcal{RM}}(\epsilon^{x'_{min}}) = \epsilon^{x'_{min}}\{p_{\mathcal{RM}}(\epsilon^{x'_{min}})/\epsilon^{x'_{min}}\}$. Because $x'_{min}\leq m'$ and  $p_{\mathcal{RM}}(\epsilon^{x'_{min}})/\epsilon^{x'_{min}} =min_{0<x\leq m'}\{p_{\mathcal{RM}}(\epsilon^x)/\epsilon^x\}$, then we have $\epsilon^{x'_{min}} \{p_{\mathcal{RM}}(\epsilon^{x'_{min}})/\epsilon^{x'_{min}}\} \leq \epsilon^{m'}\times min_{0<x\leq m'}\{p_{\mathcal{RM}}(\epsilon^x)/\epsilon^x\}=p_{\mathcal{RRM}}(\epsilon^{m'})$. That is $p_{\mathcal{RRM}}(\epsilon^m)<p_{\mathcal{RRM}}(\epsilon^{m'})$.

We prove that $p_{\mathcal{RRM}}$ satisfies the subadditive property as follows. We have $p_{\mathcal{RRM}}(\epsilon^m)/\epsilon^m=min_{0<x\leq m}\{p_{\mathcal{RM}}(\epsilon^x)/\epsilon^x\}\geq min_{0<x\leq m'}\{p_{\mathcal{RM}}(\epsilon^x)/\epsilon^x\}=p_{\mathcal{RRM}}(\epsilon^{m'})/\epsilon^{m'}$. Therefore, we have $p_{\mathcal{RRM}}(\epsilon^m)/\epsilon^m\geq p_{\mathcal{RRM}}(\epsilon^{m'})/\epsilon^{m'}$ which is the subadditive property constraint.

Let $x_{min}=\arg min_{0<x\leq m}\{p_{\mathcal{RM}}(\epsilon^x)/\epsilon^x\}$, i.e., $x_{min}\leq m$. We show that for every $m>0$, we have $p_{\mathcal{RM}}(\epsilon^m)/2\leq p_{\mathcal{RRM}}(\epsilon^m)$ as follows. We have $p_{\mathcal{RM}}(\epsilon^m)=p_{\mathcal{RM}}(\epsilon^{x_{min}}\frac{\epsilon^m}{\epsilon^{x_{min}}})\leq p_{\mathcal{RM}}(\epsilon^{x_{min}}\lceil \frac{\epsilon^m}{\epsilon^{x_{min}}}\rceil)\leq  \lceil \frac{\epsilon^m}{\epsilon^{x_{min}}}\rceil p_{\mathcal{RM}}(\epsilon^{x_{min}})$ because $p_{\mathcal{RM}}$ satisfies the subadditive property constraint. Therefore, we have $p_{\mathcal{RRM}}(\epsilon^m)=\epsilon^m\{\frac{p_{\mathcal{RM}}(\epsilon^{x_{min}})}{\epsilon^{x_{min}}}\}\geq \frac{\epsilon^m}{\epsilon^{x_{min}}}\{\frac{p_{\mathcal{RM}}(\epsilon^m)}{\lceil \frac{\epsilon^m}{\epsilon^{x_{min}}} \rceil}\}\geq \frac{\epsilon^m}{\epsilon^{x_{min}}}\{\frac{p_{\mathcal{RM}}(\epsilon^m)}{ \frac{\epsilon^m}{\epsilon^{x_{min}}} +1}\}\geq p_{\mathcal{RM}}(\epsilon^m)/2$ because $x_{min}\leq m$. Because $p_{\mathcal{RRM}}(\epsilon^m)=\epsilon^m\times min_{0<x\leq m}\{\frac{p_{\mathcal{RM}}(\epsilon^x)}{\epsilon^x}\}\leq \epsilon^m\times \{\frac{p_{\mathcal{RM}}(\epsilon^x)}{\epsilon^x}\}_{x=m}=p_{\mathcal{RM}}(\epsilon^m)$, we have $p_{\mathcal{RRM}}(\epsilon^m)\leq p_{\mathcal{RM}}(\epsilon^m)$. Therefore, for each $m>0$, we have $\sum_{k=1}^{K'} \cdot \mathbb{I}(\bm{tm}_k==m)\cdot\mathbb{I}(p_{\mathcal{RRM}}(\epsilon^m) \leq \bm{v}_k)\geq \sum_{k=1}^{K'} \cdot \mathbb{I}(\bm{tm}_k==m)\cdot\mathbb{I}(p_{\mathcal{RM}}(\epsilon^m) \leq \bm{v}_k)$. With $p_{\mathcal{RM}}(\epsilon^m)/2\leq p_{\mathcal{RRM}}(\epsilon^m)$, we conclude that $OPT(\mathcal{RM})/2 \leq OPT(\mathcal{RRM})$.
\end{proof}

\partitle{Dynamic Programming Algorithm}
In this part, we show an efficient dynamic programming algorithm to solve the relaxed revenue maximization problem.

\begin{figure}[htb]
 \centering
 \includegraphics[width=0.39\textwidth]{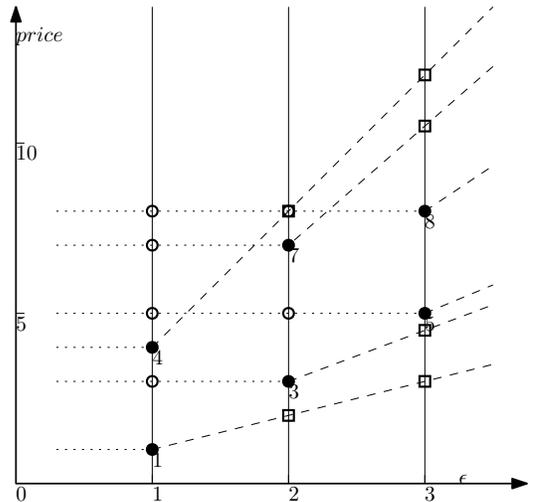}
 \vspace{-1em}
 \caption{Revenue maximization example.}
 \label{fig:rmalgorithm}
\end{figure}

At first glance, for each model, it seems that all possible values in the price range can be an optimal price, which makes the problem arguably intractable to solve. In the following, we show how to construct the complete solution space in the discrete space and prove the complete solution space is sufficient to obtain the maximum revenue.

\emph{Constructing Complete Solution Space.} It is easy to see that those survey price points should be contained in the complete solution space. For each survey price point $(m,sp^m[j])$, it determines unit price $sp^m[j]/\epsilon^m$ and price $sp^m[j]$. The general idea is that if we choose $(m,sp^m[j])$ as the price point in model $\mathcal{M}^m$, it affects the price for models $\mathcal{M}^k,k=1,...,m-1$ due to the monotone constraint and the unit price for models $\mathcal{M}^k,k=m+1,...,M$ due to the subadditive constraint. If we set the optimal price in model $\mathcal{M}^m$ as $sp^m[j]$, the unit price of the following models after model $\mathcal{M}^m$ cannot be larger than $sp^m[j]/\epsilon^m$. Therefore, for each survey price point $(m,sp^m[j])$, we draw one line $l_{(m,sp^m[j])}$ through survey price point $(m,sp^m[j])$ and the original point. For each model $\mathcal{M}^m$, we draw one vertical line $l_{\mathcal{M}^m}$. By intersecting line $l_{(m,sp^m[j])}$ and $l_{\mathcal{M}^m}$, we obtain $M-m$ new price points $(l_{\mathcal{M}^k},l_{(m,sp^m[j])})$ for $k=m+1,...,M$. We note that we do not need to generate the price points for $k=1,...,m-1$ because the unit price of model $\mathcal{M}^m$ can only constrain the unite price of model $\mathcal{M}^k, k=m+1,...,M$.  Furthermore, for each model, its price is also determined by the survey price of its right neighbors. Therefore, we need to add the survey price points of model $\mathcal{M}^m$ to models $\mathcal{M}^{k}, k=1,...,m-1$. The detailed algorithm for constructing the complete solution space is shown in Algorithm \ref{Alg:conSoluSpace}. In Lines 11-15, we use $f(m,p^m[j])$ to distinguish the survey price points from the other points in the complete solution space. For ease of presentation in the following, we name the price point in the complete solution space from Line 1 as $\mathcal{SV}$ (survey) point, the price point from Line 8 as $\mathcal{SC}$ (subadditivity constraint) point, and the price point from Line 10 as $\mathcal{MC}$ (monotonicity constraint) point.

\begin{algorithm}[thb] \caption{Constructing complete solution space for the relaxed revenue maximization problem.}\label{Alg:conSoluSpace}
\SetKwInOut{Input}{input}\SetKwInOut{Output}{output}

\Input{Model with noise parameter $\epsilon^m$ and the $j^{th}$ lowest survey price point for model with noise parameter $\epsilon^m$, denoted as $(m,sp^m[j])$.}
\Output{Complete solution space.}

add all the survey price points $(m,sp^m[j])$ to the complete solution space\;
\For{each survey price point $(m,sp^m[j])$}{
    draw a line $l_{(m,sp^m[j])}$ through this point and the original point\;
}

\For{each model with noise parameter $\epsilon^m$}{
    draw a vertical line $l_{\mathcal{M}^m}$\;
}

\For{each line $l_{\mathcal{M}^m}$}{
    \For{each line $l_{(m,sp^m[j])}$}{
        add point $(l_{\mathcal{M}^k},l_{(m,sp^m[j])})$ by intersecting line $l_{\mathcal{M}^m}$ and line $l_{(m,sp^m[j])}$ to the complete solution space for $k=m+1,...,M$\;}}

\For{each survey price point $(m,sp^m[j])$}{
    add price point $(k,sp^m[j])$ to the complete solution space for $k=1,...,m-1$\;}

\For{each price point $(m,p^m[j])$ in the complete solution space}{
    \If{$(m,p^m[j])$ is a survey price point}{
        $f(m,p^m[j])=1$\;}
    \Else{
        $f(m,p^m[j])=0$\;}}
\end{algorithm}

\begin{example}
We show an running example of Algorithm \ref{Alg:conSoluSpace}. In Figure \ref{fig:rmalgorithm}, we assume $\epsilon^1=1$, $\epsilon^2=2$, and $\epsilon^3=3$. We add the survey price points $(1,sp^1[1]=1)$, $(1,sp^1[2]=4)$, $(2,sp^2[1]=3)$, $(2,sp^2[2]=7)$, $(3,sp^3[1]=5)$, and $(3,sp^3[2]=8)$ to the complete solution space in Line 1. In Line 2, for the survey price point $(1,1)$, we draw a line $l_{(1,1)}$ through this point and the original point in Line 3. In Line 4, for model $\mathcal{M}^1$ with noise parameter $\epsilon^1=1$, we draw a vertical line $l_{\mathcal{M}^1}$. In Lines 6-8, for $l_{1,1}$ and $l_{\mathcal{M}^2}$, we add intersection $(l_{\mathcal{M}^2},l_{1,sp^1[1]})=(2,2)$ to the complete solution space. In total, we have six such new price points shown in box. In Line 9, for survey price point $(3,sp^3[1])=(3,5)$, we add price points $(2,5)$ and $(1,5)$ to the complete solution space in Line 10. Similarly, we also have six such new price points shown in circle. Therefore, for the complete solution space, we have six price points for models $\mathcal{M}^1$ and $\mathcal{M}^3$. We have five price points for model $\mathcal{M}^2$ because the intersection point of $l_{\mathcal{M}^2},l_{1,sp^1[2]}=(2,8)$ is same to the $\mathcal{MC}$ point $(k,sp^3[2])=(2,8)$ for $k=2$. In Lines 12-15, we have $f(2,3)=1$ and $f(2,5)=0$.
\end{example}

\begin{theorem}
  The complete solution space constructed by Algorithm \ref{Alg:conSoluSpace} is sufficient for finding the optimal solution of the relaxed revenue maximization problem.
\end{theorem}
\begin{proof}
As we discussed in constructing complete solution space, for each survey price point $(m,sp^m[j])$, it can affect the model revenue of model $\mathcal{M}^m$, the unit price for models $\mathcal{M}^k,k=m+1,...,M$ and the price for models $\mathcal{M}^k,k=1,...,m-1$. We prove that the $\mathcal{SC}$ and $\mathcal{MC}$ points are non-recursive, i.e., we do not need to generate new price points in the complete solution space based on the generated $\mathcal{SC}$ and $\mathcal{MC}$ points.

Given a survey price point $(m,sp^m[j])$ in model $\mathcal{M}^m$, it determines a $\mathcal{SC}$ point in model $\mathcal{M}^{m'}$, where $m'>m$, i.e., $(m',sp^m[j]/\epsilon^m\times \epsilon^{m'})$. We do not need to generate $\mathcal{SC}$ points based on $(m',sp^m[j]/\epsilon^m\times \epsilon^{m'})$ because it has the same $\mathcal{SC}$ points for models $\mathcal{M}^k,k=m'+1,...,M$ with $(m,sp^m[j])$. We may use $(m',sp^m[j]/\epsilon^m\times \epsilon^{m'})$ to generate a $\mathcal{MC}$ point $(k,sp^m[j]/\epsilon^m\times \epsilon^{m'})$. If $k>m$, the new point $(k,sp^m[j]/\epsilon^m\times \epsilon^{m'})$ is not necessary because $sp^m[j]/\epsilon^m\times \epsilon^{m'}/\epsilon^k>sp^m[j]/\epsilon^m$ which violates the subadditive constraint. If $k<m$, the new point $(k,sp^m[j]/\epsilon^m\times \epsilon^{m'})$ is also not necessary because $sp^m[j]/\epsilon^m\times \epsilon^{m'}>sp^m[j]$ which violates the subadditive constraint.

Given a survey price point $(m,sp^m[j])$ in model $\mathcal{M}^m$, it determines a $\mathcal{MC}$ point in model $\mathcal{M}^{m'}$, where $m'<m$, i.e., $(m',sp^m[j])$. We do not need to generate $\mathcal{MC}$ points based on $(m',sp^m[j])$ because those $\mathcal{MC}$ points are already determined by $(m,sp^m[j])$. It is also not necessary to generate $\mathcal{SC}$ points based on $(m',sp^m[j])$ because if $(m',sp^m[j])$ and $(m,sp^m[j])$ are chosen as the optimal prices, all the optimal prices for models $\mathcal{M}^k,k=m',...m,...,M$ are determined, i.e., $p(\epsilon^{m'})=...=p(\epsilon^{m})=...p(\epsilon^{M})$.
\end{proof}

\emph{A recursive solution.}
We define the revenue of an optimal solution recursively in terms of the optimal solutions to subproblems. We pick as our subproblems the problems of determining the maximum revenue $OPT(m,j)$, where $OPT(m,j)$ denotes the maximum revenue for considering the first $m$ models and taking the $j^{th}$ lowest price point in the complete solution space of model $\mathcal{M}^m$. For the full problem, the maximum revenue would be $\max\{OPT(M,j)\}$ for all the price points $(M,p^M[j])$ in the complete solution space in model $\mathcal{M}^M$. For the price points in the complete solution space of model $\mathcal{M}^1$, we can directly compute $OPT(1,j)$ for all the price points because there is no initial constraint. For the price points in the complete solution space of other models, we need to consider both the monotone constraint and the subadditive constraint. We have a recursive equation as follows.
\begin{equation}\label{equ:recursive}
OPT(m,j)=\max\{OPT(m-1,j')\}+MR(m,j)
\end{equation}
where $p^{m-1}[j']\leq p^m[j]\&\&p^{m-1}[j']/\epsilon^{m-1}\geq p^m[j]/\epsilon^m$ and $MR(m,j)$ denotes the revenue from model $\mathcal{M}^m$ if we price model $\mathcal{M}^m$ for $p^m[j]$.

\emph{Computing the maximum revenue.}
Now, we could easily write a recursive algorithm in Algorithm \ref{Alg:DpRRM} based on recurrence (\ref{equ:recursive}), where $|p^m|$ is the number of the price points in model $\mathcal{M}^m$.

\begin{algorithm}[thb] \caption{Dynamic programming algorithm for finding an optimal solution of the relaxed revenue maximization problem.}\label{Alg:DpRRM}
\SetKwInOut{Input}{input}\SetKwInOut{Output}{output}

\Input{Model with noise parameter $\epsilon^m$ and its corresponding price points in the complete solution space.}
\Output{$OPT(\mathcal{RRM})$.}

\For{each model $\mathcal{M}^m$}{
sort the price points in the complete solution space in decreasing order\;
use $(m,p^m[j])$ to denote the $j^{th}$ lowest price point\;
    $MR(m,|p^m|)=p^m[|p^m|]f(m,p^m[|p^m|])$\;
    \For{$j=|p^m|-1$ to $1$}{
        $MR(m,j)=p^m[j]\sum_{k=j}^{|p^m|}f(m,p^m[k])$\;}}

\For{$j=1$ to $|p^1|$}{
    $OPT(1,j)=MR(1,p^1[j])$\;}

\For{each model $\mathcal{M}^m, m=2,...,M$}{
    \For{each price point $(m,p^m[j])$}{
       $OPT(m,j)=\max\{OPT(m-1,j')\}+MR(m,j)$, where $p^{m-1}[j']\leq p^m[j]\&\&p^{m-1}[j']/\epsilon^{m-1}\geq p^m[j]/\epsilon^m$\;
       $p(L.OPT(m,j))=p^{m-1}[j']$ that satisfies $OPT(m,j)$ in Line 11\;}}

$OPT(RRM)=\max\{OPT(M,j)\}$, where $j=1$ to $|p^m|$.
\end{algorithm}

\begin{theorem}
  Algorithm \ref{Alg:DpRRM} can be finished in $O(N^2M^2)$ time.
\end{theorem}
\begin{proof}
For the $O(N)$ survey price points, we general $O(NM)$ price points in the complete solution space. For each price point in the complete solution space, we need $O(NM)$ time to update $OPT(m,j)$. Therefore, Algorithm \ref{Alg:DpRRM} requires $O(N^2M^2)$ time.
\end{proof}

\begin{example}
In model $\mathcal{M}^1$ of Figure \ref{fig:rmalgorithm}, we have $OPT(1,j)$ for $j=1,2,...,6$ shown in Table \ref{tab:solExam}. For computing $OPT(2,2)$, there is only one price point $(1,3)$ satisfying both the monotone constraint and the subadditive constraint within model $\mathcal{M}^1$. Therefore, we have $OPT(2,2)=OPT(1,2)+MR(2,2)=3+6=9$. Similarly, we can fill the entire table shown in Table \ref{tab:solExam}.
\end{example}

\emph{Constructing an optimal solution.}
Although Algorithm \ref{Alg:DpRRM} determines the maximum revenue of $\mathcal{RRM}$, it does not directly show the optimal price for each model $p(\epsilon^m)$. However, for each price point $(m,p^m[j])$ in the complete solution space, we record the price point $p(L.OPT(m,j))$ in model $\mathcal{M}^{m-1}$ which has the maximum revenue in those price points that satisfy both the monotone constraint and the subadditive constraint with $(m,p^m[j])$ in Line 12 of Algorithm \ref{Alg:DpRRM}. Therefore, we can recursively backtrack the optimal price point in model $\mathcal{M}^{m-1}$ from the optimal price point in model $\mathcal{M}^m$. We need $O(nM)$ time to find the maximum value in $OPT(M,j)$ and $O(M)$ time to backtrack. Therefore, we can construct an optimal solution in $O(NM)$ time. We note that such a solution may be one of several solutions that can achieve the optimal value.

We show an running example in Table \ref{tab:solExam}. We first obtain $OPT(3,3)$ which has the maximum value among $OPT(3,j)$ for $j=1,2,...,6$. Therefore, we set $p(\epsilon^3)=5$. We backtrack to $OPT(2,2)$ in model $\mathcal{M}^2$ and set $p(\epsilon^2)=3$. And then we backtrack to $OPT(1,2)$ in model $\mathcal{M}^1$ and set $p(\epsilon^1)=3$. Finally, an optimal pricing setting is $\langle p(\epsilon^1),p(\epsilon^2),p(\epsilon^3)\rangle=\langle 3,3,5\rangle$. We note that in our running example, the pricing setting $\langle 4,5,5\rangle$ also has the maximum revenue $19$.

\begin{table}[htb]\centering
\caption{Example for constructing an optimal solution.}\label{tab:solExam}
 \vspace{-1em}
{%
\begin{tabular}{|c|c|c|c|c|c|c|}
\hline
\diagbox{Model}{\# of points} & 1 & 2 & 3 & 4 & 5 & 6\\
\hline
$\mathcal{M}^1$ & 2  & 3  & 4  & 0    & 0    & 0\\
\hline
$\mathcal{M}^2$ & 6  & 9  & 9  & 11   & 4    & Null\\
\hline
$\mathcal{M}^3$ & 15 & 18 & 19 & 19   & 11   & 4\\
\hline
\end{tabular}}
\end{table}%

\begin{assumption}
The broker is honest but curious.
\end{assumption}

\begin{remark}
In practice, the broker may be semi-honest or even malicious. For these cases, we can take advantage of the local differential privacy \cite{DBLP:conf/focs/DuchiJW13} in which the data owners and the model buyers can add DP noise by themselves before sending their data to the broker. Also, encryption-based techniques can be further incorporated into the market design.
\end{remark}

\subsection{Complete Differentially Private Data Marketplace Dynamics}\label{sub:DPcomplete}
We summarize the differentially private data marketplace with model-based pricing dynamics from the broker's perspective, which is an end-to-end data marketplace with practical considerations and consists of computationally efficient component algorithms. The detailed algorithm is shown in Algorithm \ref{Alg:complete.broker}, which integrates all the proposed algorithms in the previous sections. We assume that the broker can set appropriate parameters $M$, $\epsilon^m$, and $\mathcal{MB}^m$ based on her market experiences, which is reasonable. For example, Microsoft can easily determine the different features assigned to Windows 10 Home version and Windows 10 Pro version. Finally, although instantiated with the DP market, we stress that Algorithm \ref{Alg:complete.broker} can be applied to the general setting by switching the DP parameter to the risk factor.


\begin{algorithm}[thb] \caption{The complete broker functioning in DP-\emph{Dealer} Pipeline. }\label{Alg:complete.broker}
\SetKwInOut{Input}{input}\SetKwInOut{Output}{output}


collect data and usage restriction among $n$ data owners: collect dataset $\bm{D}=\{\bm{z}_i\}$ along with DP restriction parameters $\epsilon_i$ and extra compensation function $\bm{ec}_i$ for $i=1,2,...,n$\;

decide a set of $M$ models to train with model privacy parameter $\epsilon^{m}$ and manufacturing budget $\mathcal{MB}^m$ for $m=1,2,...,M$\;

\%\% Model training and releasing\;

\For{m=1 to M}{

data valuation: call Algorithm \ref{Alg:MCShapley} to compute approximate Shapley value $\mathcal{SV}_i^m$ for $i=1,2,...,n$\;

base compensation: compute $\bm{bc}_i^m=\frac{\mathcal{SV}_i^m}{\sum_{i=1}^{n}\mathcal{SV}_i^m}\mathcal{MB}^m$\;

data selection: call Algorithm \ref{Alg:DynamicProg}, \ref{Alg:greedy}, or \ref{Alg:GuessGreedy} to select training subset $\bm{S}^m$ with manufacturing budget $\mathcal{MB}^m$ to maximize $\mathcal{SV}(\bm{S}^m)$\;

model training: train the model with subset $\bm{S}^m$ by Algorithm \ref{Alg:DPtraining}\;

model releasing: release model $\mathcal{M}^m$, its pricing $p(\epsilon^m)$ and estimated excess population loss $\Delta\mathcal{L}(\mathcal{A}^m_{Algorithm\ref{Alg:DPtraining}};\epsilon^m,\bm{S}^m)$\;

}

perform market survey among $K'$ sampling model buyers (survey participants): collect market survey results of model demand $\bm{tm}_k$ and valuation $v_k$ for $k=1,2,...,K'$\;

\%\% Model pricing\;

model pricing: call Algorithms \ref{Alg:conSoluSpace} and \ref{Alg:DpRRM} to compute the optimal price $p(\epsilon^m)$ of model $\mathcal{M}^m$ for $m=1,2,...,M$\;

\%\% Compensation allocation\;
\For{m=1 to M}{
compute $\bm{bc}_i$ and $\bm{ec}_i$ of data owner $\mathcal{O}^i$ for $i\in\bm{S}^m$ by proportionally dividing $\frac{p(\epsilon^m)}{\sum_{m=1}^M p(\epsilon^m)}OPT(RRM)$ and allocate the corresponding compensation to $\mathcal{O}_i$\;
}

\end{algorithm}


\section{Experiments}\label{sec:Experiments}
In this section, we present experimental studies validating: 1)our proposed mechanisms for compensation allocation are efficient and effective; 2) our proposed mechanisms for pricing models can generate more revenue for the data owners and the broker; 3) our exquisitely designed dynamic programming algorithms for pricing models significantly outperform the baseline algorithms.

\subsection{Experiment Setup}
We ran experiments on a machine with an Intel Core i7-8700K and two NVIDIA GeForce GTX 1080 Ti running Ubuntu with 64GB memory. We employed SVM classifier as our model and used both synthetic datasets and a real Breast Cancer dataset \cite{Dua:2019} in our experiments. We implemented the following algorithms in Matlab 2018a.

\begin{itemize}
\item \textbf{Greedy}: The greedy algorithm for compensation allocation in Algorithm \ref{Alg:greedy}.
\item \textbf{PPDP}: The pseudo-polynomial dynamic programming algorithm for compensation allocation in Algorithm \ref{Alg:DynamicProg}.
\item \textbf{GuessGreedy}: The guess and greedy algorithm for compensation allocation in Algorithm \ref{Alg:GuessGreedy}.
\item \textbf{Dealer}: The optimal prices computed by the dynamic programming algorithm in Algorithm \ref{Alg:DpRRM} with survey price space.
\item \textbf{Dealer+}: The optimal prices computed by the dynamic programming algorithm in Algorithm \ref{Alg:DpRRM} with complete solution space.
\item \textbf{Linear}: We take the lowest survey price from model $\mathcal{M}^1$ and the highest survey price from model $\mathcal{M}^M$ and use linear interpolation for the remaining models $\mathcal{M}^2,...,\mathcal{M}^{M-1}$ based on the two end-prices.
\item \textbf{Low}: We set the lowest price from all survey prices to all models.
\item \textbf{Median}: We set the median price from all survey prices to all models.
\item \textbf{High}: We set the highest price from all survey prices to all models.
\end{itemize}

\subsection{Compensation Allocation}
Figures \ref{fig:allocation}(a)(b) shows the compensation allocation time cost and accuracy of Greedy, PPDP, and GuessGreedy of a various number of data owners, respectively. Because PPDP is significantly affected by the budget and cannot work for a very large budget, we set all budget to $10000$ in our experiments for fair comparison. Figure \ref{fig:allocation}(a) shows the time cost for a various number of data owners. Greedy significantly outperforms both PPDP and GuessGreedy due to its simplicity. GuessGreedy costs the highest time cost because we need to enumerate $n\choose v$ subsets, where $n$ is the total number of data owners and $v$ is the size of the sampled subsets during enumeration. In our experiments, the time cost for GuessGreedy is prohibitively high even we set $v=2$. We skip some results of Greedy and PPDP in the figures due to their prohibitively high time cost. Figure \ref{fig:allocation}(b) shows the accuracy of various algorithms on five models. We employ three different algorithms Greedy, PPDP, and GuessGreedy to choose three subsets for each manufacturing budget $\mathcal{MB}^m=1.2\sum_{i=1}^n \mathcal{SV}_i, 1.4\sum_{i=1}^n \mathcal{SV}_i, 1.6\sum_{i=1}^n \mathcal{SV}_i, 1.8\sum_{i=1}^n \mathcal{SV}_i$, and 2.0$\sum_{i=1}^n \mathcal{SV}_i$, respectively. We also use ALL as a baseline, which includes all the patients. We add differential privacy with parameters $\epsilon=0.01, 0.1, 1, 5$, and $10$ in the training processing (Algorithm \ref{Alg:DPtraining}) on the four datasets, respectively. We can see that although the number of patients in subsets selected by Greedy, PPDP, and GuessGreedy are less than the number of patients in ALL, the accuracy is higher for larger $\epsilon$, which verifies the effectiveness of Shapely value. For example, for $\epsilon=10$, there are only 337 patients in the subset selected by Greedy, but the accuracy on that subset is higher than the entire dataset ALL. For smaller $\epsilon$, the accuracy on ALL is higher. The reason is that for smaller $\epsilon$, with the less budget, we obtain a smaller subset. For example, for $\epsilon=0.01$, there are only 259 patients in the subset selected by Greedy. Comparing different algorithms, the accuracy of the subset selected by PPDP is only a little higher than Greedy. Therefore, we can employ Greedy for most cases.

\begin{figure}[!htb]
\centering
\subfigure[time cost.]{
\begin{minipage}[b]{0.22\textwidth}
\includegraphics[width=1.1\textwidth]{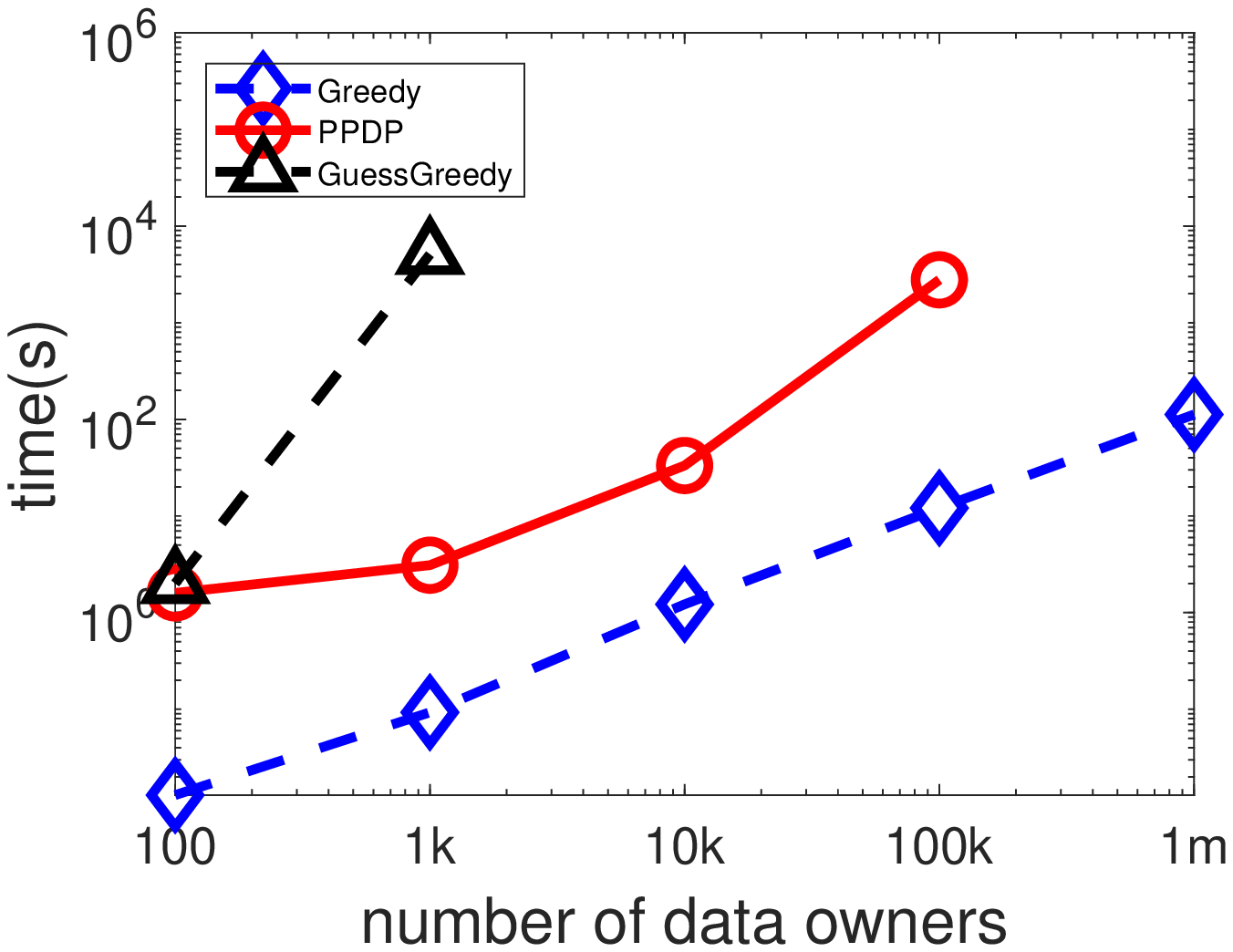}
\end{minipage}
}
\subfigure[accuracy.]{
\begin{minipage}[b]{0.22\textwidth}
\includegraphics[width=1.1\textwidth]{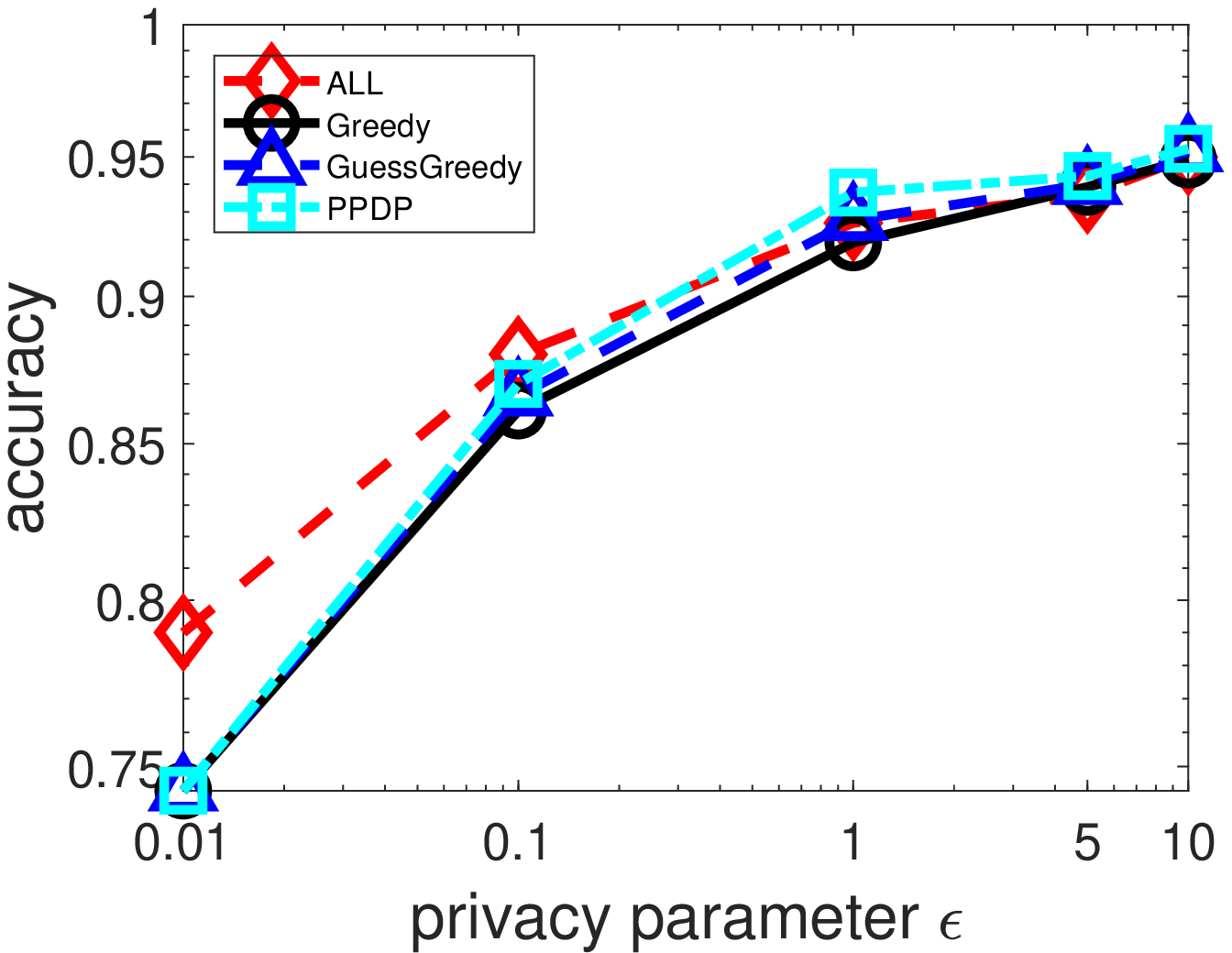}
\end{minipage}
}
 \vspace{-1em}\caption{Compensation allocation.} \label{fig:allocation}
\end{figure}

\subsection{Revenue Maximization of Pricing Models}
We experimentally study the revenue gain of our proposed algorithms on different distributed datasets. We generate two datasets with $100$ survey price points (i.e., collecting from $100$ potential model buyers). The number of survey price points on each model follows independent random distribution and Gaussian (mean=5, standard deviation=3) random distribution, respectively. For the first model of both datasets, we generate those survey price points following independent distribution with range $[1000,5000]$. The remaining nine models follow a $100$ increase on both the lower and upper bound of the range.

Figures \ref{fig:inde}(a)(b)(c)(d) show the dataset, price, affordability ratio (fraction of the model buyers that can afford to buy a model), and revenue on an independent random distributed dataset, respectively. Figure \ref{fig:inde}(b) shows that Dealer+, Dealer, and Linear have a similar price setting distribution. All models in Dealer have different prices. For the price setting distribution of Dealer+, the first model and the second model have the same price, the same to the fifth model and the sixth model, the eighth model and the ninth model, which maximizes the revenue comparing to Dealer and verifies the effectiveness of our complete solution space construction. Figure \ref{fig:inde}(c) shows that Dealer+ has the highest affordability ratio except for Low. For the most critical metric revenue, Dealer+ outperforms the other algorithms at least $10\%$, which verifies the gain of our complete solution space construction.

In practical applications, it is more likely that the survey price point datasets follow Gaussian distribution rather than independent random distribution. Figures \ref{fig:gaussian}(a)(b)(c)(d) show that all algorithms have similar performances on Gaussian distributed dataset as on independent distributed dataset.

\begin{figure*}[t]
\centering
\subfigure[Data distribution]{
\begin{minipage}[b]{0.22\textwidth}
\includegraphics[width=1.15\textwidth]{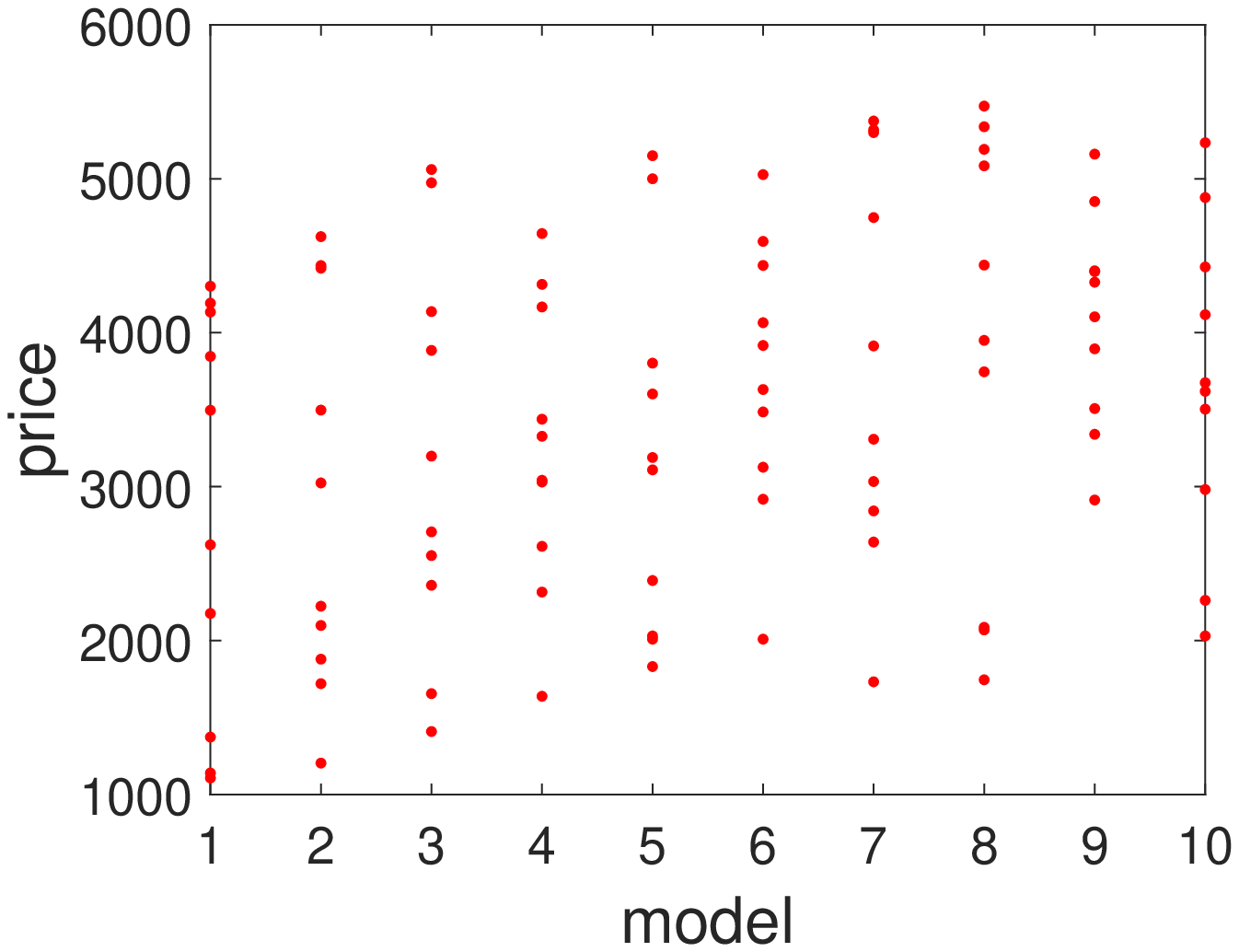}
\end{minipage}
}
\subfigure[Price]{
\begin{minipage}[b]{0.22\textwidth}
\includegraphics[width=1.15\textwidth]{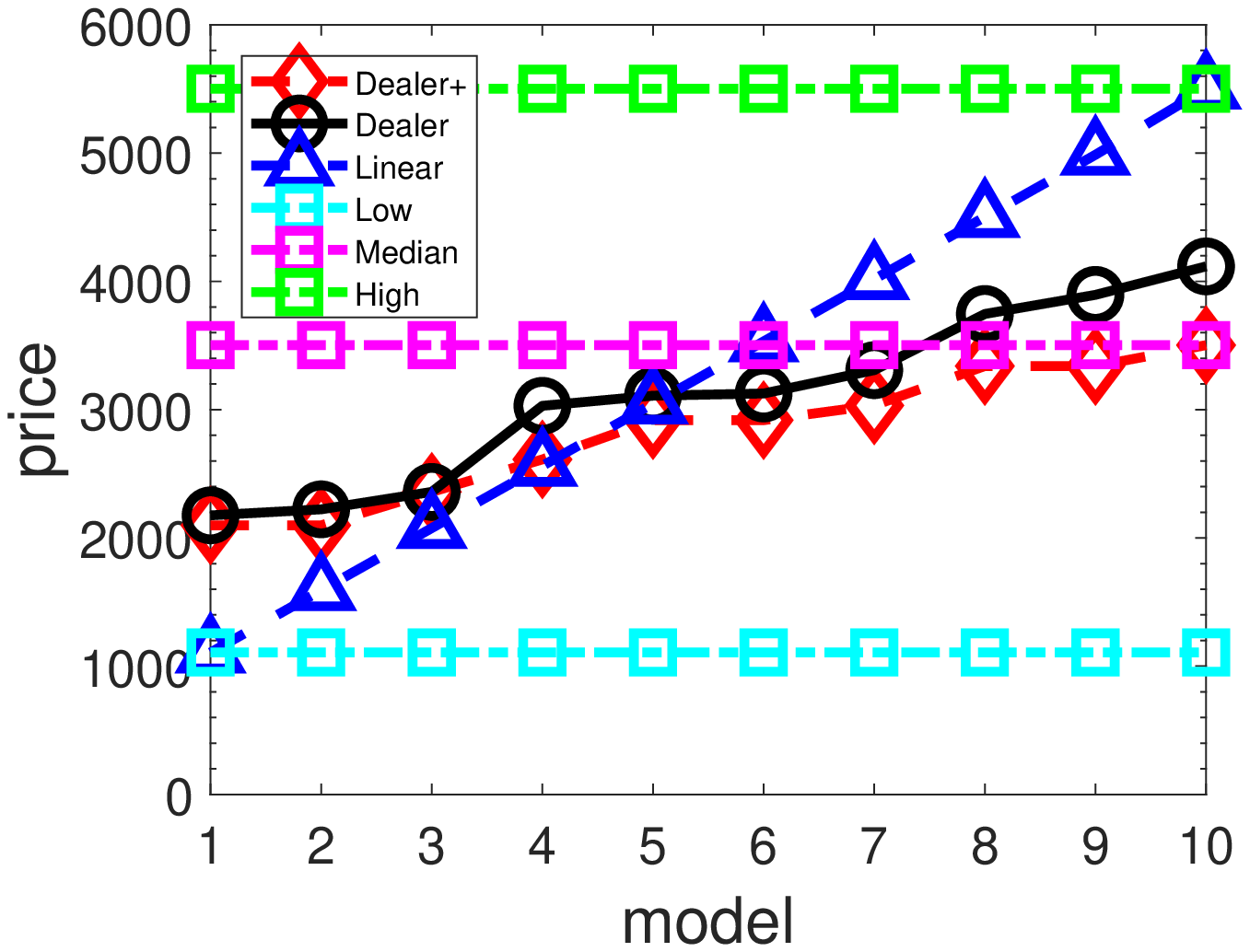}
\end{minipage}
}
\subfigure[Ratio]{
\begin{minipage}[b]{0.22\textwidth}
\includegraphics[width=1.15\textwidth]{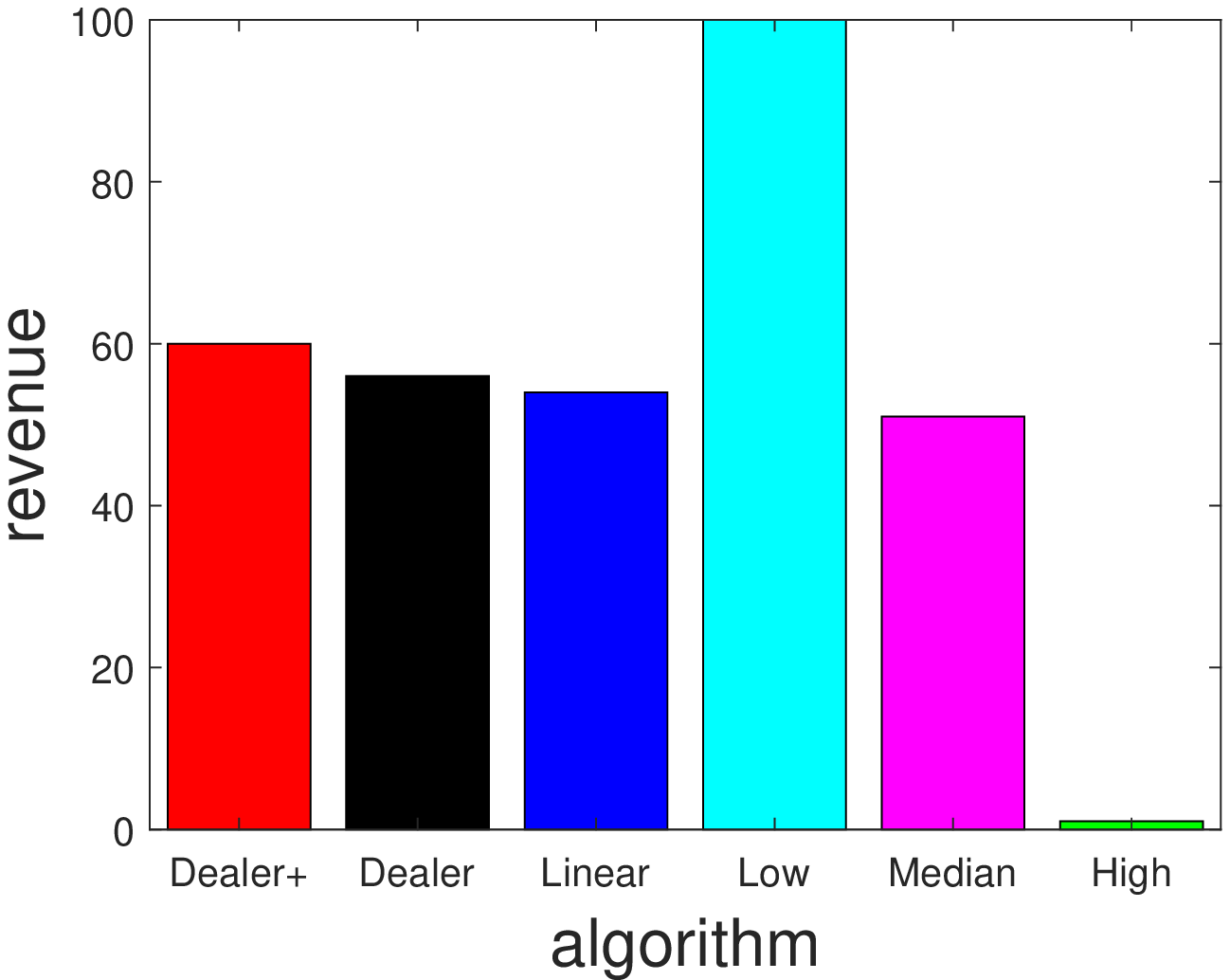}
\end{minipage}
}
\subfigure[Revenue]{
\begin{minipage}[b]{0.22\textwidth}
\includegraphics[width=1.15\textwidth]{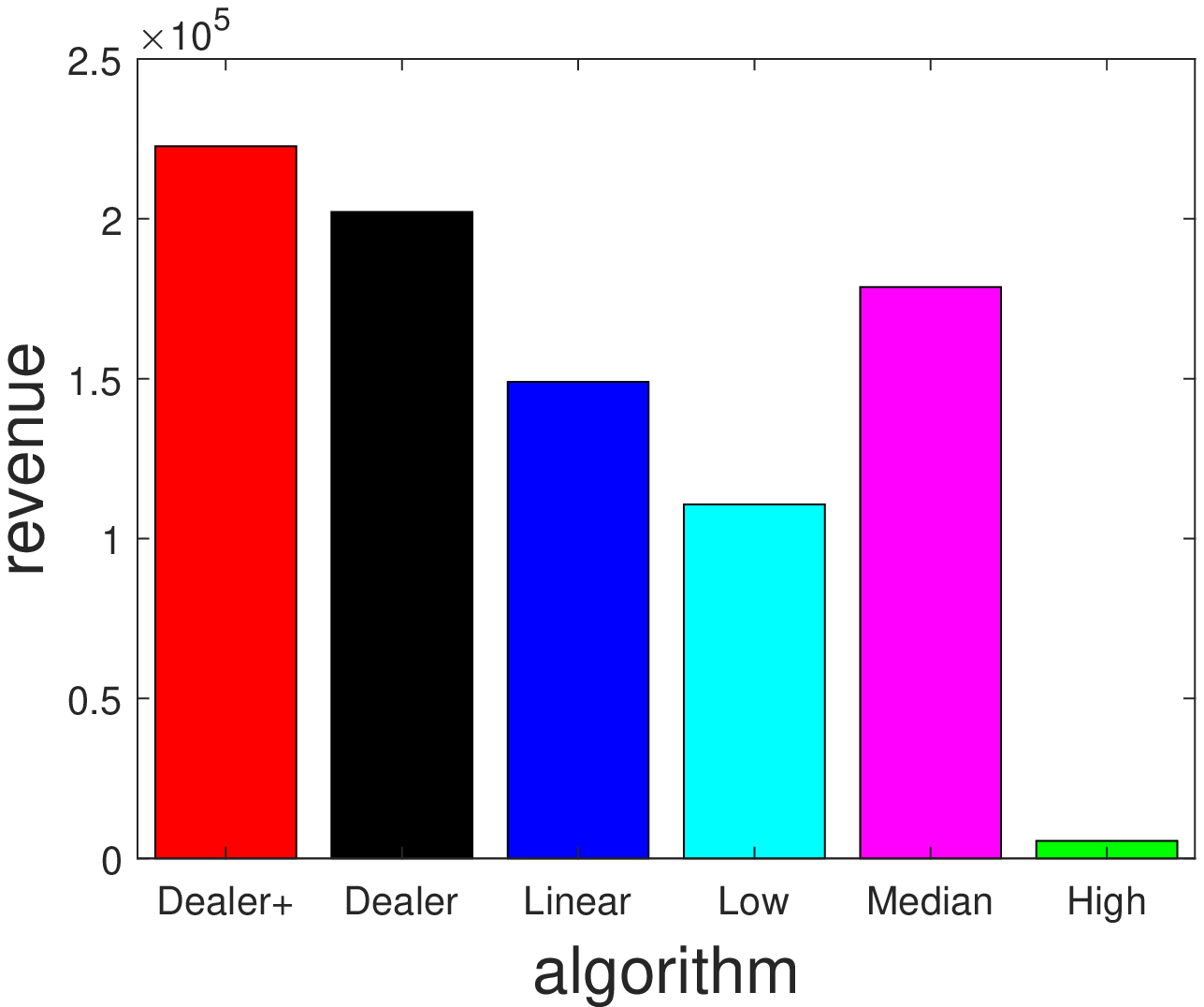}
\end{minipage}
}
 \vspace{-1em}\caption{Independent distribution.
 } \label{fig:inde}
\end{figure*}

\begin{figure*}[t]
\centering
\subfigure[Data distribution]{
\begin{minipage}[b]{0.22\textwidth}
\includegraphics[width=1.15\textwidth]{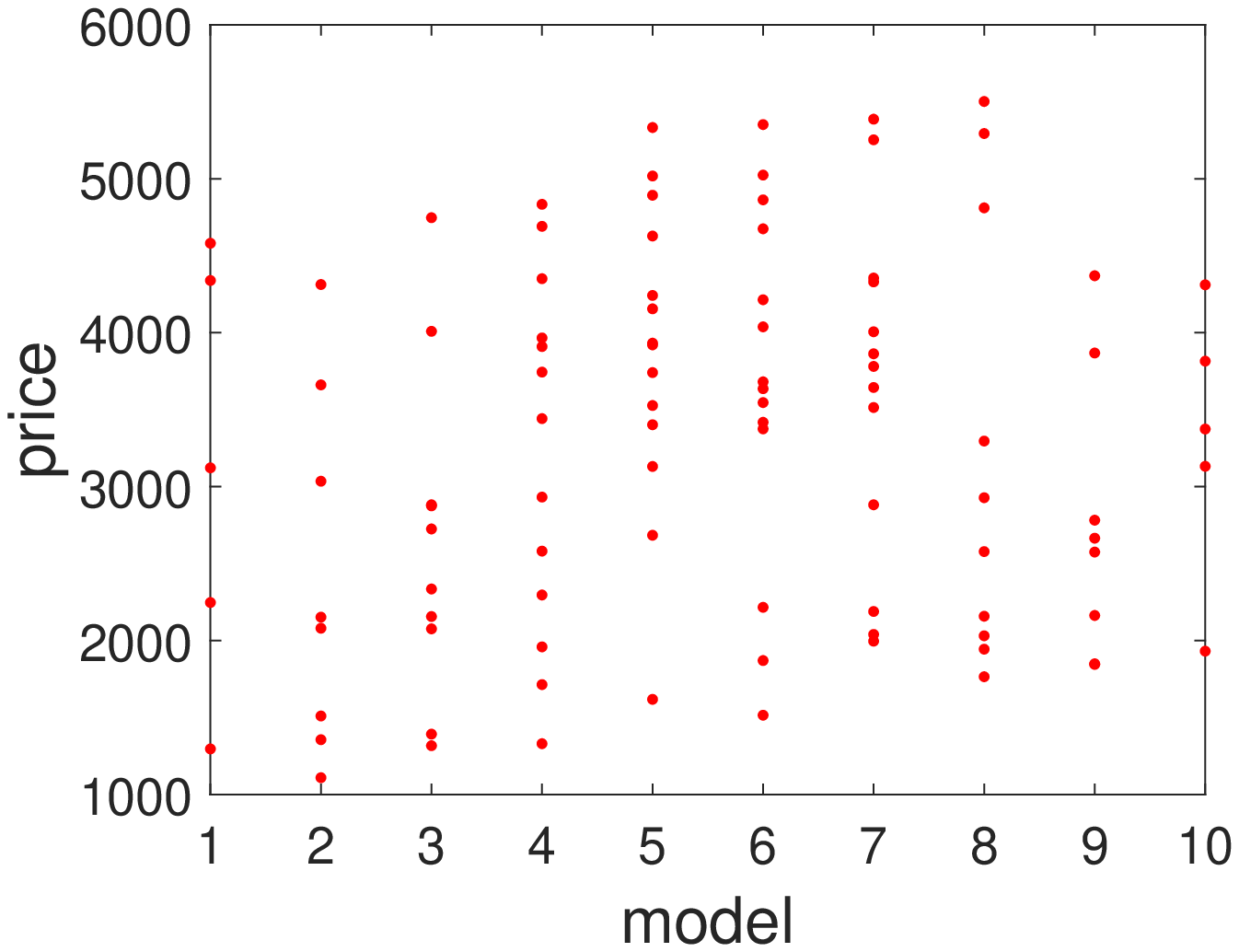}
\end{minipage}
}
\subfigure[Price]{
\begin{minipage}[b]{0.22\textwidth}
\includegraphics[width=1.15\textwidth]{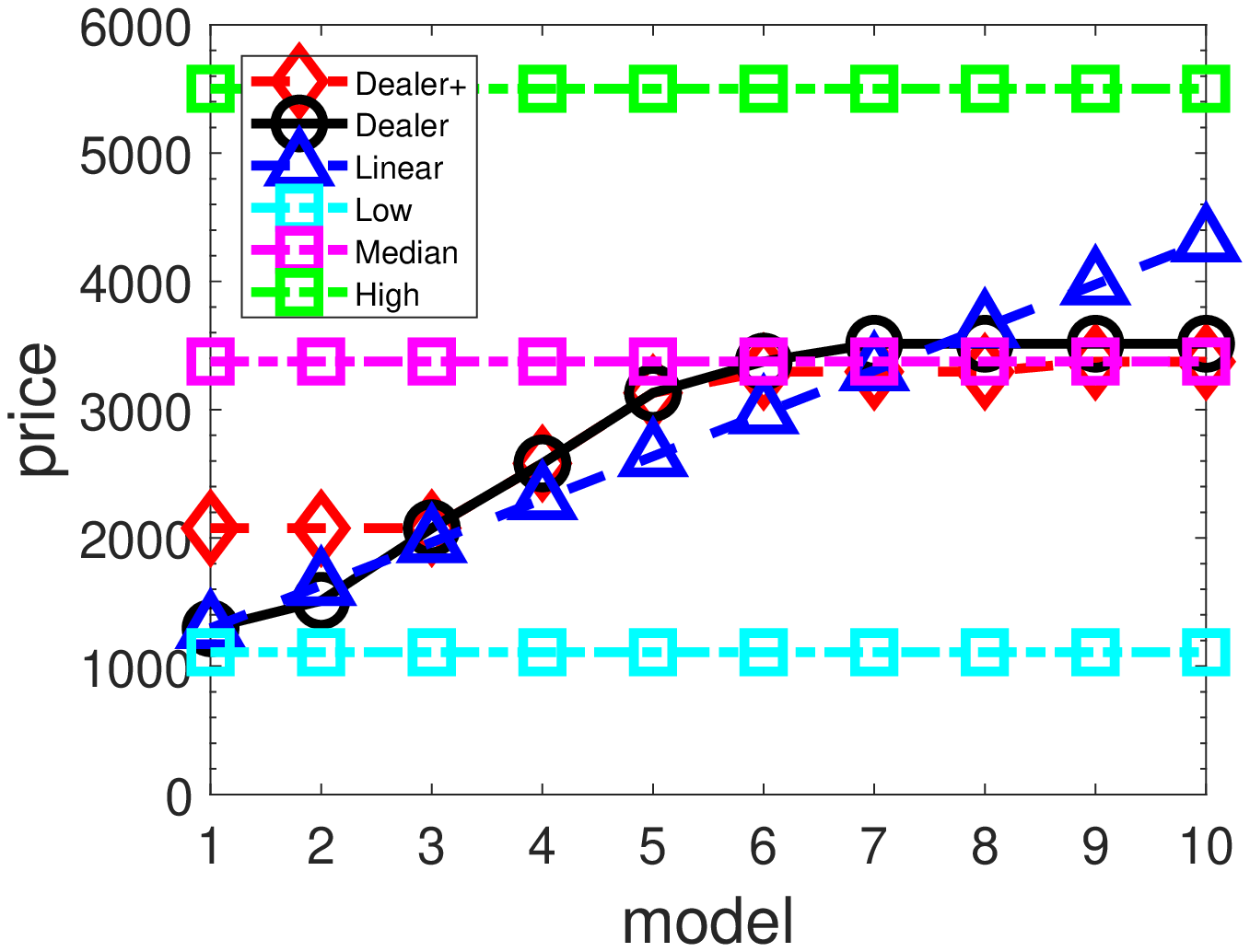}
\end{minipage}
}
\subfigure[Ratio]{
\begin{minipage}[b]{0.22\textwidth}
\includegraphics[width=1.15\textwidth]{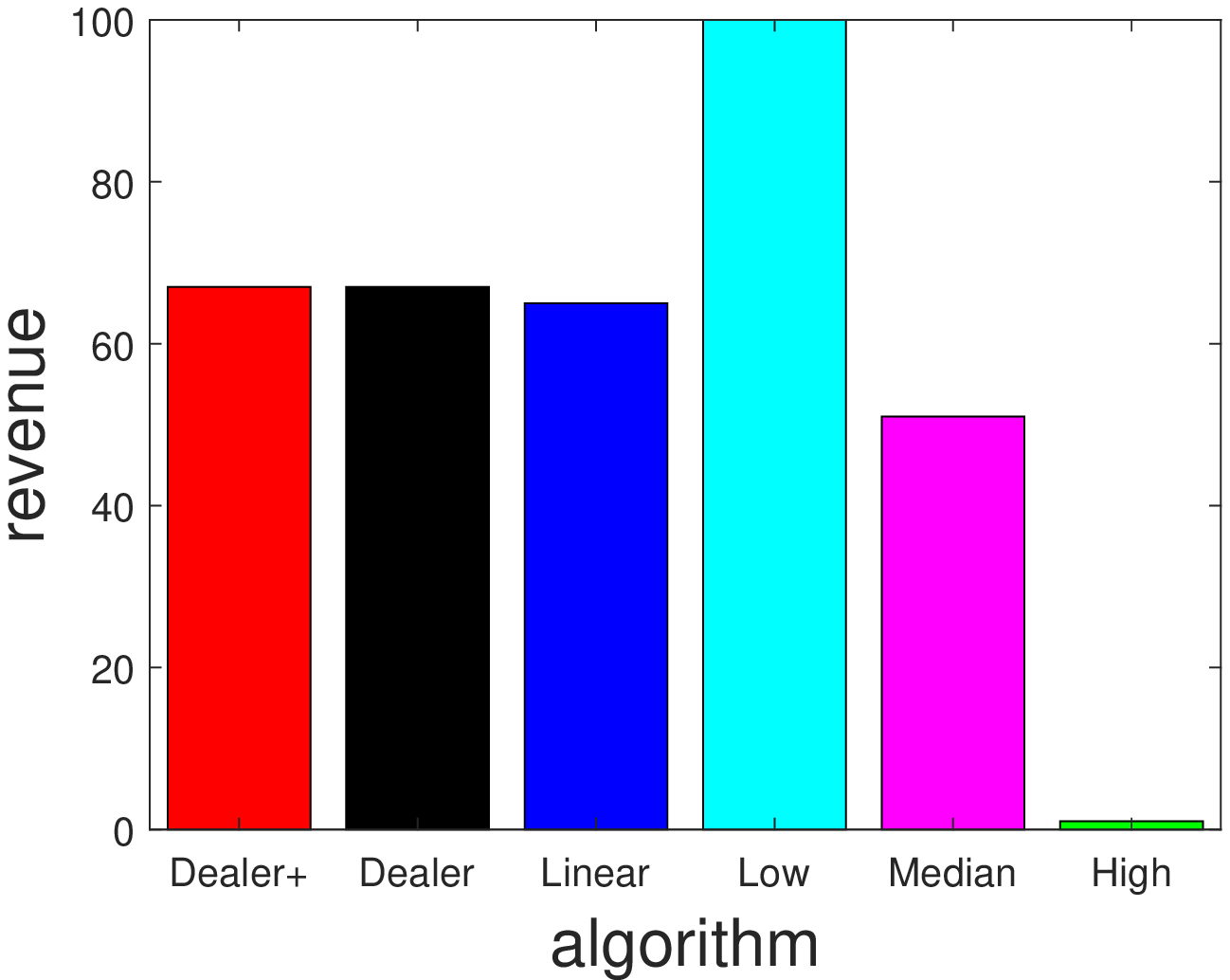}
\end{minipage}
}
\subfigure[Revenue]{
\begin{minipage}[b]{0.22\textwidth}
\includegraphics[width=1.15\textwidth]{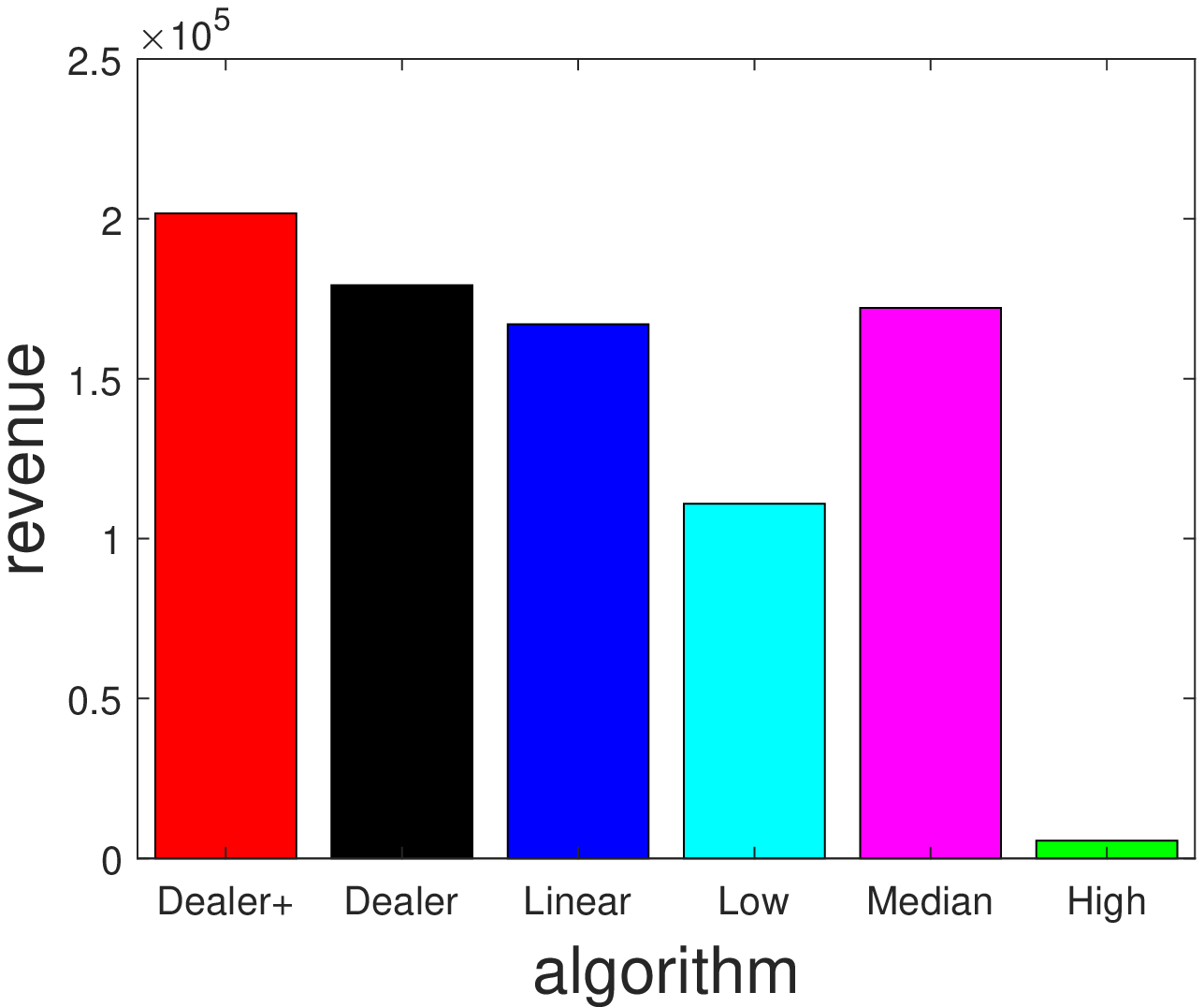}
\end{minipage}
}
 \vspace{-1em}\caption{Gaussian distribution.
 } \label{fig:gaussian}
\end{figure*}

\subsection{Time Cost of Pricing Models}
We experimentally study the efficiency of our proposed algorithms for pricing models. Because Linear, Low, Median, and High algorithms only need to scan through the survey price points once, the time cost is low. For the ease of presentation, we omit the experimental results for those four algorithms. Instead, we compare our proposed Dealer and Dealer+ with the classic exhaustion based approach. We first apply exhaustion-based approach to our complete solution space named Base. However, the time cost of most of the experiments is prohibitively high. Therefore, we apply exhaustion-based approach to the survey price space named BaseAppr.

\begin{figure}[htb]
 \centering
 \includegraphics[width=0.5\textwidth]{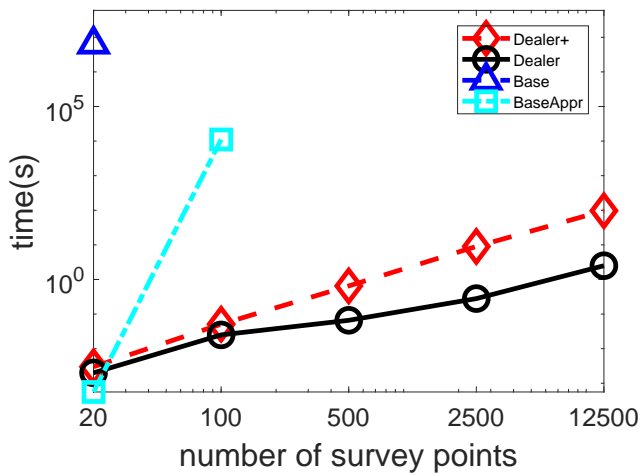}
 \vspace{-1em}
 \caption{Time cost.}
 \label{fig:timecost}
\end{figure}

Figure \ref{fig:timecost} shows the time cost of Dealer, Dealer+, Base, and BaseAppr on a various number of survey price points. Both Dealer+ and Dealer linearly increase with the increase of the number of survey price points, which verifies the efficiency of our proposed dynamic programming algorithm. In practical applications, $12500$ survey points are enough for most of the surveys, the optimal Dealer+ only requires dozens of seconds on a PC. If the time cost is very sensitive, we can employ Dealer which searches from the original survey price space with a slight tradeoff in optimal revenue. The time cost of both Base and BaseAppr is prohibitively high due to the high volume price combinations for different models.


\section{Conclusion and Future Work}\label{sec:Conclusion}
In this paper, we proposed the first end-to-end data marketplace with model-based pricing framework towards answering the question: how can the broker assign value to the data owners based on their contribution to the models to incentivize more data contributions, and determine optimal prices for a series of models for various model buyers to maximize the revenue with arbitrage-free guarantee. For the former, we introduced a Shapley value-based mechanism to quantify each data owner's value towards all the models trained out of the contributed data and the data owners have the abilities to control their data usage. For the latter, we designed a pricing mechanism based on the models' privacy parameter to maximize the revenue. We proposed Gen-\emph{Dealer} to model the end-to-end data marketplace with model-based pricing and illustrate a concrete realization of differentially private data marketplace with model-based pricing DP-\emph{Dealer} which provably satisfies the desired formal properties. Extensive experiments verified that DP-\emph{Dealer} is efficient.

There are several exciting directions for future work. First, multiple brokers can co-exist in practical applications, which form a competitive relationship to maximize the revenue for themselves. Second, multiple risk factors forming a risk vector can be considered, which enables the market to take care of different types of demands of the data owners. Third, personalized model manufacturing can be considered, which tailors the model training to each model buyer to best suit their budget and model usage scenario.


\bibliographystyle{abbrv}
\bibliography{CDP}

\begin{thebibliography}{10}

\bibitem{dawex}
Dawex, https://www.dawex.com/en/.

\bibitem{googlequery}
Google bigquery, https://cloud.google.com/bigquery/.

\bibitem{Twitter}
https://support.gnip.com/apis/.

\bibitem{Bloomberg}
https://www.bloomberg.com/professional/product/market-data/.

\bibitem{IOTA}
Iota, https://data.iota.org/.

\bibitem{agarwal2019marketplace}
A.~Agarwal, M.~Dahleh, and T.~Sarkar.
\newblock A marketplace for data: An algorithmic solution.
\newblock In {\em Proceedings of the 2019 ACM Conference on Economics and
  Computation}, pages 701--726. ACM, 2019.

\bibitem{DBLP:journals/teco/AlaeiMS14}
S.~Alaei, A.~Malekian, and A.~Srinivasan.
\newblock On random sampling auctions for digital goods.
\newblock {\em {ACM} Trans. Economics and Comput.}, 2(3):11:1--11:19, 2014.

\bibitem{DBLP:conf/icml/AnconaOG19}
M.~Ancona, C.~{\"{O}}ztireli, and M.~H. Gross.
\newblock Explaining deep neural networks with a polynomial time algorithm for
  shapley value approximation.
\newblock In {\em Proceedings of the 36th International Conference on Machine
  Learning, {ICML} 2019, 9-15 June 2019, Long Beach, California, {USA}}, pages
  272--281, 2019.

\bibitem{bassily2019private}
R.~Bassily, V.~Feldman, K.~Talwar, and A.~G. Thakurta.
\newblock Private stochastic convex optimization with optimal rates.
\newblock In {\em Advances in Neural Information Processing Systems}, pages
  11279--11288, 2019.

\bibitem{DBLP:journals/cor/CastroGT09}
J.~Castro, D.~G{\'{o}}mez, and J.~Tejada.
\newblock Polynomial calculation of the shapley value based on sampling.
\newblock {\em Computers {\&} {OR}}, 36(5):1726--1730, 2009.

\bibitem{DBLP:journals/pr/CawleyT03LOO}
G.~C. Cawley and N.~L.~C. Talbot.
\newblock Efficient leave-one-out cross-validation of kernel fisher
  discriminant classifiers.
\newblock {\em Pattern Recognition}, 36(11):2585--2592, 2003.

\bibitem{DBLP:journals/pvldb/ChawlaDKT19}
S.~Chawla, S.~Deep, P.~Koutris, and Y.~Teng.
\newblock Revenue maximization for query pricing.
\newblock {\em {PVLDB}}, 13(1):1--14, 2019.

\bibitem{DBLP:conf/sigmod/ChenK019}
L.~Chen, P.~Koutris, and A.~Kumar.
\newblock Towards model-based pricing for machine learning in a data
  marketplace.
\newblock In {\em Proceedings of the 2019 International Conference on
  Management of Data, {SIGMOD} Conference 2019, Amsterdam, The Netherlands,
  June 30 - July 5, 2019.}, pages 1535--1552, 2019.

\bibitem{Dua:2019}
D.~Dua and C.~Graff.
\newblock {UCI} machine learning repository, 2017.

\bibitem{DBLP:conf/focs/DuchiJW13}
J.~C. Duchi, M.~I. Jordan, and M.~J. Wainwright.
\newblock Local privacy and statistical minimax rates.
\newblock In {\em 54th Annual {IEEE} Symposium on Foundations of Computer
  Science, {FOCS} 2013, 26-29 October, 2013, Berkeley, CA, {USA}}, pages
  429--438, 2013.

\bibitem{dwork2006calibrating}
C.~Dwork, F.~McSherry, K.~Nissim, and A.~Smith.
\newblock Calibrating noise to sensitivity in private data analysis.
\newblock In {\em Theory of cryptography conference}, pages 265--284. Springer,
  2006.

\bibitem{dwork2014algorithmic}
C.~Dwork, A.~Roth, et~al.
\newblock The algorithmic foundations of differential privacy.
\newblock {\em Foundations and Trends{\textregistered} in Theoretical Computer
  Science}, 9(3--4):211--407, 2014.

\bibitem{DBLP:journals/ai/FatimaWJ08}
S.~S. Fatima, M.~J. Wooldridge, and N.~R. Jennings.
\newblock A linear approximation method for the shapley value.
\newblock {\em Artif. Intell.}, 172(14):1673--1699, 2008.

\bibitem{DBLP:conf/icml/GhorbaniZ19}
A.~Ghorbani and J.~Y. Zou.
\newblock Data shapley: Equitable valuation of data for machine learning.
\newblock In {\em Proceedings of the 36th International Conference on Machine
  Learning, {ICML} 2019, 9-15 June 2019, Long Beach, California, {USA}}, pages
  2242--2251, 2019.

\bibitem{DBLP:conf/sigecom/GhoshR11}
A.~Ghosh and A.~Roth.
\newblock Selling privacy at auction.
\newblock In {\em Proceedings 12th {ACM} Conference on Electronic Commerce
  (EC-2011), San Jose, CA, USA, June 5-9, 2011}, pages 199--208, 2011.

\bibitem{DBLP:conf/soda/GuruswamiHKKKM05}
V.~Guruswami, J.~D. Hartline, A.~R. Karlin, D.~Kempe, C.~Kenyon, and
  F.~McSherry.
\newblock On profit-maximizing envy-free pricing.
\newblock In {\em Proceedings of the Sixteenth Annual {ACM-SIAM} Symposium on
  Discrete Algorithms, {SODA} 2005, Vancouver, British Columbia, Canada,
  January 23-25, 2005}, pages 1164--1173, 2005.

\bibitem{iyengar2019towards}
R.~Iyengar, J.~P. Near, D.~Song, O.~Thakkar, A.~Thakurta, and L.~Wang.
\newblock Towards practical differentially private convex optimization.
\newblock In {\em IEEE S and P}.

\bibitem{DBLP:conf/uss/Jayaraman019}
B.~Jayaraman and D.~Evans.
\newblock Evaluating differentially private machine learning in practice.
\newblock In {\em 28th {USENIX} Security Symposium, {USENIX} Security 2019,
  Santa Clara, CA, USA, August 14-16, 2019}, pages 1895--1912, 2019.

\bibitem{jia2019efficient}
R.~Jia, D.~Dao, B.~Wang, F.~A. Hubis, N.~M. Gurel, B.~Li, C.~Zhang, C.~Spanos,
  and D.~Song.
\newblock Efficient task-specific data valuation for nearest neighbor
  algorithms.
\newblock {\em Proceedings of the VLDB Endowment}, 12(11):1610--1623, 2019.

\bibitem{koutris2012query}
P.~Koutris, P.~Upadhyaya, M.~Balazinska, B.~Howe, and D.~Suciu.
\newblock Query-based data pricing.
\newblock In {\em Proceedings of the 31st ACM SIGMOD-SIGACT-SIGAI symposium on
  Principles of Database Systems}, pages 167--178. ACM, 2012.

\bibitem{koutris2013toward}
P.~Koutris, P.~Upadhyaya, M.~Balazinska, B.~Howe, and D.~Suciu.
\newblock Toward practical query pricing with querymarket.
\newblock In {\em proceedings of the 2013 ACM SIGMOD international conference
  on management of data}, pages 613--624. ACM, 2013.

\bibitem{DBLP:journals/jacm/KoutrisUBHS15}
P.~Koutris, P.~Upadhyaya, M.~Balazinska, B.~Howe, and D.~Suciu.
\newblock Query-based data pricing.
\newblock {\em J. {ACM}}, 62(5):43:1--43:44, 2015.

\bibitem{DBLP:conf/icdt/LiLMS13}
C.~Li, D.~Y. Li, G.~Miklau, and D.~Suciu.
\newblock A theory of pricing private data.
\newblock In {\em Joint 2013 {EDBT/ICDT} Conferences, {ICDT} '13 Proceedings,
  Genoa, Italy, March 18-22, 2013}, pages 33--44, 2013.

\bibitem{DBLP:journals/tods/LiLMS14}
C.~Li, D.~Y. Li, G.~Miklau, and D.~Suciu.
\newblock A theory of pricing private data.
\newblock {\em {ACM} Trans. Database Syst.}, 39(4):34:1--34:28, 2014.

\bibitem{DBLP:journals/cacm/LiLMS17}
C.~Li, D.~Y. Li, G.~Miklau, and D.~Suciu.
\newblock A theory of pricing private data.
\newblock {\em Commun. {ACM}}, 60(12):79--86, 2017.

\bibitem{DBLP:journals/pvldb/LinK14}
B.~Lin and D.~Kifer.
\newblock On arbitrage-free pricing for general data queries.
\newblock {\em {PVLDB}}, 7(9):757--768, 2014.

\bibitem{shapiro1998versioning}
C.~Shapiro and H.~Varian.
\newblock Versioning: The smart way to sell information.
\newblock {\em Harvard Business Review}, 76(6):107--115, 1998.

\bibitem{shapley1953value}
L.~S. Shapley.
\newblock A value for n-person games.
\newblock {\em Contributions to the Theory of Games}, 2(28):307--317, 1953.

\end{thebibliography}

\end{document}